\documentclass[sigconf,authorversion]{aamas} 

\usepackage{balance} 
\usepackage[labelfont={bf}]{subcaption}
\usepackage{makecell}
\setcellgapes{1pt}
\usepackage{algorithm}
\usepackage{algpseudocode}
\usepackage{xcolor}
\usepackage{amsthm}
\newtheorem{proposition}{Proposition}
\newtheorem{remark}{Remark}

\usepackage{listings}
\usepackage[skins]{tcolorbox}
\definecolor{ShadowColor}{RGB}{30,150,190}
\definecolor{codegreen}{rgb}{0,0.6,0}
\definecolor{codegray}{rgb}{0.5,0.5,0.5}
\definecolor{codepurple}{rgb}{0.58,0,0.82}
\definecolor{backcolour}{rgb}{0.95,0.95,0.92}
\lstdefinestyle{mystyle}{
    backgroundcolor=\color{backcolour},   
    commentstyle=\color{codegreen},
    keywordstyle=\color{magenta},
    numberstyle=\tiny\color{codegray},
    stringstyle=\color{codepurple},
    basicstyle=\ttfamily\tiny,
    breakatwhitespace=false,         
    breaklines=true,                 
    captionpos=b,                    
    keepspaces=false,                 
    numbers=none,                    
    numbersep=5pt,                  
    showspaces=false,                
    showstringspaces=false,
    showtabs=false,                  
    tabsize=1
}

\makeatother

\makeatletter
\gdef\@copyrightpermission{
  \begin{minipage}{0.2\columnwidth}
   \href{https://creativecommons.org/licenses/by/4.0/}{\includegraphics[width=0.90\textwidth]{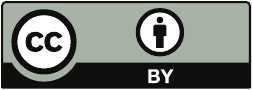}}
  \end{minipage}\hfill
  \begin{minipage}{0.8\columnwidth}
   \href{https://creativecommons.org/licenses/by/4.0/}{This work is licensed under a Creative Commons Attribution International 4.0 License.}
  \end{minipage}
  \vspace{5pt}
}
\makeatother

\setcopyright{ifaamas}
\acmConference[AAMAS '26]{Proc.\@ of the 25th International Conference
on Autonomous Agents and Multiagent Systems (AAMAS 2026)}{May 25 -- 29, 2026}
{Paphos, Cyprus}{C.~Amato, L.~Dennis, V.~Mascardi, J.~Thangarajah (eds.)}
\copyrightyear{2026}
\acmYear{2026}
\acmDOI{}
\acmPrice{}
\acmISBN{}

\makeatletter
\let\@acmBadgeL@image\@undefined
\let\@acmBadgeR@image\@undefined
\newcommand{\@acmBadgeL@image}{}
\newcommand{\@acmBadgeR@image}{}
\makeatother

\title[Building Large-Scale Drone Defenses]{Building Large-Scale Drone Defenses\\ from Small-Team Strategies}

\author{Grant Douglas}
\affiliation{
  \institution{Adelaide University}
  \city{Adelaide}
  \country{Australia}}
\email{grant.douglas@adelaide.edu.au}

\author{Stephen Franklin}
\affiliation{
  \institution{Saab Australia Pty Ltd}
  \city{Adelaide}
  \country{Australia}}

\author{Claudia Szabo}
\affiliation{
  \institution{Adelaide University}
  \city{Adelaide}
  \country{Australia}}
\email{claudia.szabo@adelaide.edu.au}

\author{Mingyu Guo}
\affiliation{
  \institution{Adelaide University}
  \city{Adelaide}
  \country{Australia}}
\email{mingyu.guo@adelaide.edu.au}

\begin{abstract}
Defending against large adversarial drone swarms requires coordination methods that scale effectively beyond conventional multi-agent optimisation. In this paper, we propose to scale strategies proven effective in small defender teams by integrating them as modular components of larger forces using our proposed framework. A dynamic programming (DP) decomposition assembles these components into large teams in polynomial time, enabling efficient construction of scalable defenses without exhaustive evaluation. Because a unit that is strong in isolation may not remain strong when combined, we sample across multiple small-team candidates. Our framework iterates between evaluating large-team outcomes and refining the pool of modular components, allowing convergence on increasingly effective strategies. Experiments demonstrate that this partitioning approach scales to substantially larger scenarios while preserving effectiveness and revealing cooperative behaviours that direct optimisation cannot reliably discover.
\end{abstract}


\newcommand{\BibTeX}{\rm B\kern-.05em{\sc i\kern-.025em b}\kern-.08em\TeX}
\begin{document}

\pagestyle{fancy}
\fancyhead{}

\maketitle 


\section{Introduction}

Coordinating large teams of autonomous drones in adversarial environments is a formidable challenge. Even with modest numbers of agents, the combinatorial explosion of possible strategies quickly renders exhaustive search and naive optimisation infeasible. Reinforcement learning (RL) and skill-discovery methods have shown promise for small teams (e.g.,~\cite{RN80,RN191,RN276,RN274}), but they often scale poorly due to high sample complexity, non-stationarity, sparse rewards or computational overhead \cite{RN80}.

In this paper, we present a staged pipeline for building large-scale drone defenses from small team strategies. The key idea is to first identify and evolve promising behaviours in small defender teams, and then systematically reuses these strategies as modular components for larger forces. Dynamic programming (DP) is used to assemble these components in polynomial time, transforming what would otherwise be a super-exponential search into a tractable optimisation~\cite{RN460}. The resulting solution enables the efficient construction of coordinated swarm defenses that scale to realistic adversarial scenarios. While prior UAV swarm defense studies generally evaluate on no more than 12 agents (e.g., 10v10 in \cite{RN91} and 5v5 in \cite{RN89}), we demonstrate scalability to attacker swarms up to thirty agents and defender swarms up to 45 agents with stable execution, substantially exceeding the tested ranges of existing learning-based methods.

Our approach deliberately uses heuristics rather than direct policy learning. Heuristics offer several advantages critical for scalability: (i) they require limited manual scripting, enabling construction of diverse policies relatively quickly, (ii) can be expressed in compact mathematical forms for efficient JAX implementation~\cite{jax}, (iii) they can be simulated at scale without the overhead of RL training, and (iv) at runtime, execution only requires lightweight CPU-level evaluation of heuristic code - in contrast to neural network inference - enabling deployment on smaller onboard processors with modest memory. Prior work has explored heuristic evolution in multi-agent systems, including modular evolution~\cite{RN458,RN459} and messy genetic algorithms~\cite{goldberg1989messy}, but we adapt these ideas to swarm defense by framing heuristics as sub-chromosomes in evolutionary search.

The contributions of this work are both methodological and engineering-oriented. We present a hybrid GA–DP framework leveraging chromosome factorisation for efficient evolutionary search and demonstrate its scalability on large-scale drone defense simulations. Implemented fully in JAX and deployed on high-performance computing infrastructure, the framework integrates algorithmic structuring that exploits priors and hierarchical decomposition to avoid exhaustive search while remaining tractable, even across billions of simulation time steps. To expand the behavioural search space, we additionally use large language models to generate and diversify candidate heuristics. Together, these elements enable experiments at scales that would otherwise be computationally prohibitive, positioning the framework as a practical foundation for future swarm defense research and applications.

The primary contributions of this work are fourfold:
\begin{enumerate}
    \item A hybrid GA–DP framework for scalable defense. We propose a framework that combines genetic algorithms to evolve effective chromosomes in small-scale engagements with dynamic programming to assemble them into strategies for larger swarms. This allows tractable analysis at scales that naive evolutionary methods cannot reach~\cite{RN458,RN460}.
    \item DP integration with a hierarchical defense policy. By embedding DP-guided allocations within a hierarchy, our method coordinates defenders at the sub-team level while retaining heterogeneity at the agent level, enabling effective pursuit–evasion dynamics in complex swarm encounters.
    \item A chromosome factorisation perspective on evolutionary search. We show that the effective building blocks in this domain are entire chromosomes (combinations of heuristics) rather than individual genes, highlighting the importance of preserving high-order dependencies~\cite{goldberg1989messy}.
    \item Iterative refinement through large-scale evaluation. By feeding back the outcomes of large-scale evaluations into the small-team heuristics, our approach progressively improves the accuracy of heuristic estimates and converges toward high-performing defender strategies at scale. This iterative design provides solutions to the problem, as well as a proven methodology for their systematic improvement. 
\end{enumerate}

We demonstrate that this combined GA-DP approach scales to large defensive swarms, uncovering cooperative strategies that are difficult to evolve with conventional representations. Collectively, these contributions provide a practical and scalable pathway for optimizing coordinated behavior in adversarial multi-agent systems.

\section{Related Work}

We survey prior research on multi-agent coordination and strategy optimisation, highlighting approaches that inform our hierarchical heuristic framework.

\textit{Multi-Agent Defense and Pursuit–Evasion.}
Defending a target from coordinated attackers is a classic multi-agent control problem, often studied through pursuit–evasion games and adversarial swarm interactions. Algebraic and control-theoretic formulations yield insights for small engagements \cite{RN206}, but rapidly become intractable as team sizes increase \cite{RN177}. Reinforcement learning approaches \cite{RN80} can scale somewhat further, but face challenges of non-stationarity, sparse rewards, and high sample complexity, especially in competitive–cooperative domains such as drone defense \cite{RN386}. These limitations motivate methods that exploit structure or heuristics for tractability.

\textit{Evolutionary Algorithms in Multi-Agent Systems.}
Evolutionary computation offers a population-based search mechanism suited to the combinatorial complexity of multi-agent coordination. Genetic algorithms (GAs) have been applied to tasks such as role assignment, cooperative foraging, communication, and swarm robotics \cite{RN458,RN459}. However, scalability remains an issue, as the search space grows combinatorially with the number of agents. The messy genetic algorithm (mGA) \cite{goldberg1989messy} sought to preserve high-quality building blocks under deceptive landscapes, but its stochastic recombination remains expensive at scale. Our approach takes inspiration from these ideas but emphasises explicit chromosome factorisation and systematic reuse rather than probabilistic recombination.

\textit{Dynamic Programming (DP) for Team Optimisation.}
DP has been widely studied for coalition formation, task allocation, and sensor coverage in MAS \cite{RN460}. DP leverages problem structure to find exact or near-optimal decompositions, but faces exponential blow-up when applied directly to large swarms. Hybrid approaches that combine DP with other mechanisms are therefore attractive, as they can preserve structure while avoiding combinatorial explosion.

\textit{Dynamic Target Assignment.}
Classical weapon–target assignment (WTA) focuses on the discrete assignment of weapons to static targets to maximise expected damage or minimise surviving threats, typically solved via integer programming or dynamic programming \cite{Cai2006,KLINE2019226}. In contrast, our framework addresses continuously moving attackers with kinematic constraints, where assignments must account for multi-agent interactions. Furthermore, whereas DP in classical WTA computes exact or relaxed optimal solutions \cite{Ahner2013}, we employ DP as a combinatorial decomposition to guide evolutionary search, structuring chromosome sampling without requiring full optimality. This approach bridges discrete assignment principles with continuous, dynamic multi-agent engagements, enabling scalable strategy optimisation.

\textit{Hierarchical Policies in Swarm Control.}
Hierarchical organisation is a natural way to manage complexity in MAS. High-level policies assign roles or allocate agents to sub-problems, while low-level policies execute primitive behaviours such as pursuit or patrol. This paradigm has been explored both in reinforcement learning \cite{RN273,RN247,RN230} and in hybrid heuristic–learning frameworks. Our approach falls into this class but differs in two respects: (i) we employ heuristics as the low-level policies, chosen for their largely automatic generation enabling diversity and efficient JAX implementation, and (ii) we embed DP-based allocation within the hierarchy. Consequently, they facilitate coordination to emerge from structured assignments, enabling scalability without requiring each agent to learn full joint coordination from scratch.

To our knowledge, no prior work has combined evolutionary search, DP, and heuristic factorisation within a hierarchical framework for drone defense. This integration enables tractable scaling from small evolved strategies to large swarm engagements, filling a gap left by existing MARL, GA, and DP-only approaches.

\section{Approach}

\subsection{Problem Formulation and Setting}

We formalise the drone defense scenario as a multi-agent coordination problem, detailing agent dynamics, objectives, and interactions, and introduce our hierarchical heuristics solved via a hybrid genetic algorithm and dynamic programming. The environment consists of an open two-dimensional battlefield with a fixed critical asset (target) that the Blue team must defend (see Figure~\ref{fig:launchExample}). A Red team of attackers attempts to reach the target from randomised starting positions (radius and angle to target) following sinusoidal paths (randomised amplitudes and frequencies), introducing nonlinear and unpredictable motion that prevents defenders from relying on simple linear extrapolation. While not fully realistic, this design creates an informative testbed for evaluating hierarchical and heuristic coordination strategies. The objective of the Blue team is to maximise win rate - by preventing Red agents from reaching the asset - in contrast to much of the prior work, which often focuses on agent rewards \cite{RN91} or attrition \cite{RN199,RN89}.

\begin{figure}
\centering
\begin{subfigure}{.2\textwidth}
  \includegraphics[width=\linewidth]{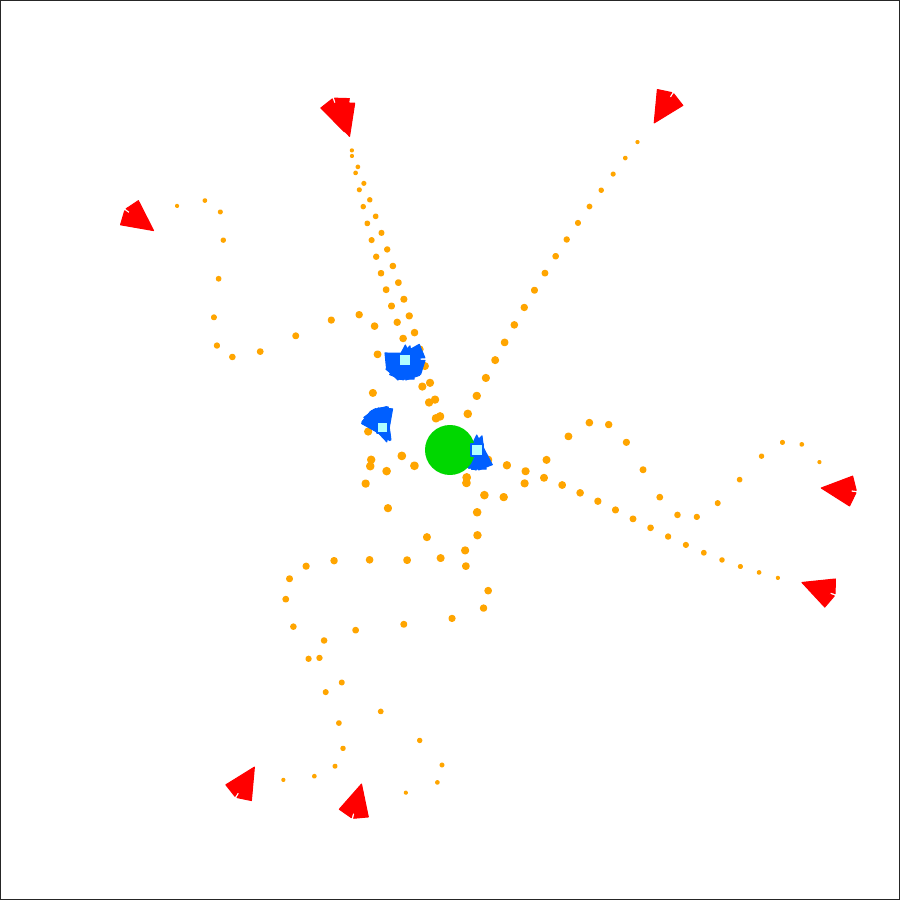}
  \caption{Launch phase}
\end{subfigure}\hspace{3mm}
\begin{subfigure}{.2\textwidth}
  \includegraphics[width=\linewidth]{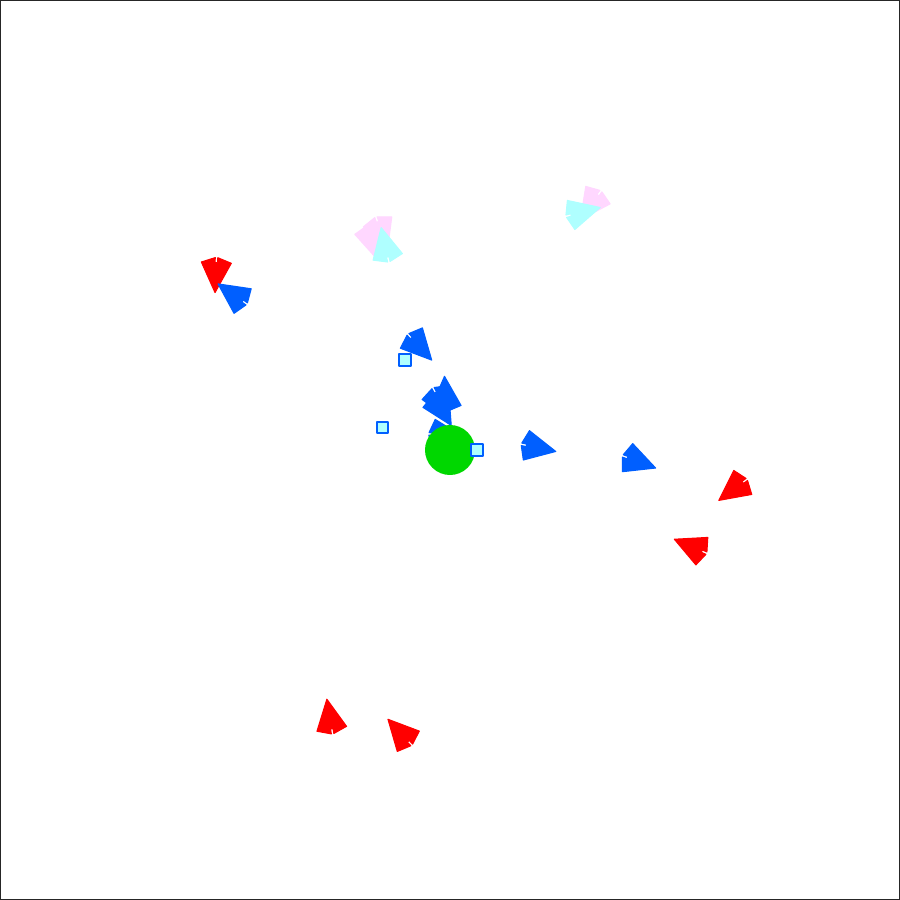}
  \caption{Mid-episode}
\end{subfigure}\hspace{3mm}
\caption{Scenario with 10 defenders (blue) protecting a target (green) from 8 attackers (red). In (a) attackers are shown with their offensive paths (orange) toward the target. In (b) collided drones that no longer participate have muted coloring.}
\label{fig:launchExample}
\end{figure}

\begin{figure*}[h]
  \centering
  \includegraphics[width=0.70\linewidth]{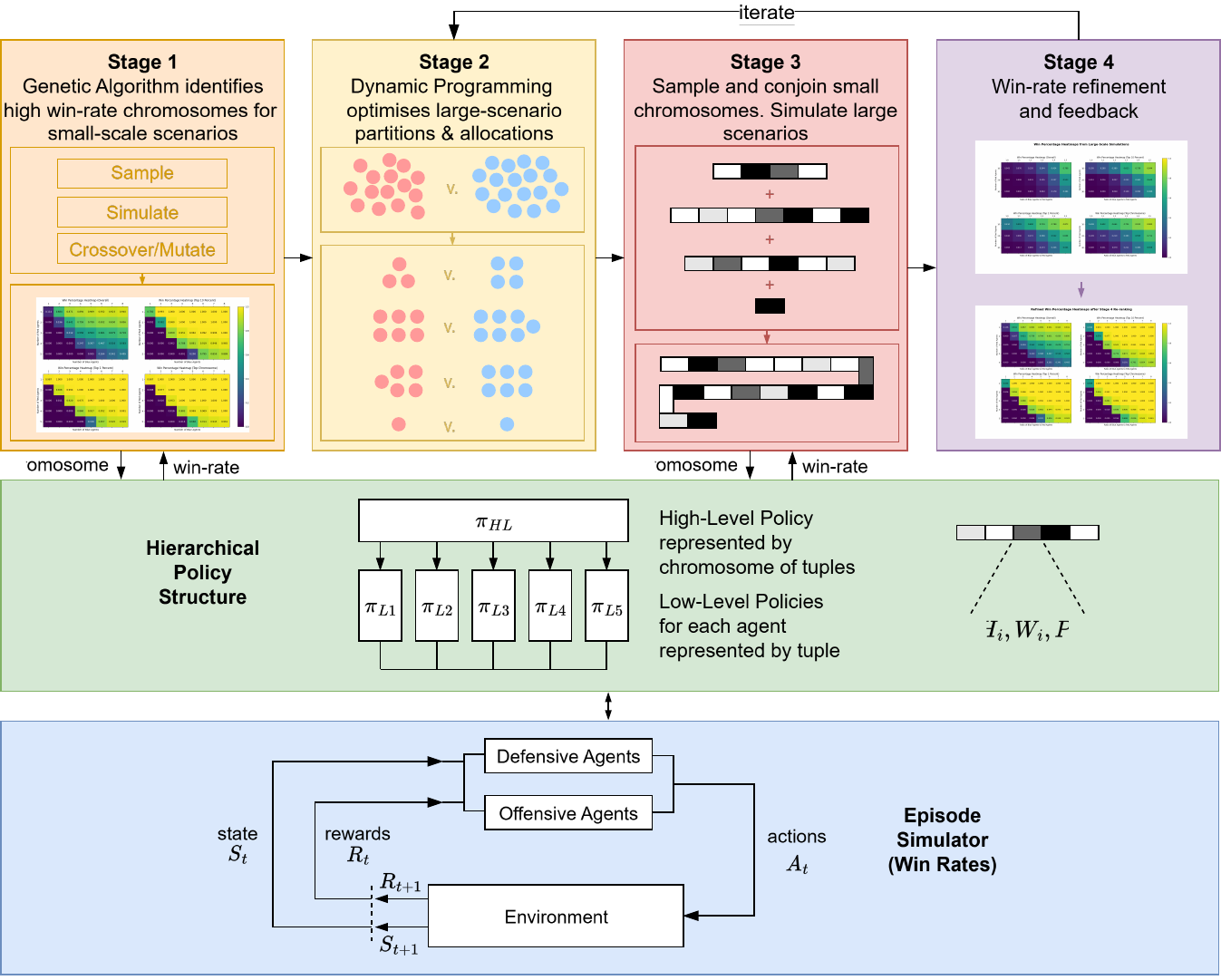}
  \caption{Overall hybrid GA–DP framework. The four-stage pipeline (Stage 1: GA evolution, Stage 2: dynamic programming allocation, Stage 3: chromosome sampling, Stage 4: refinement) operates over a hierarchical policy structure assigning heuristics to agents. These heuristics drive all decision-making. The simulator follows a standard MDP cycle, but no reinforcement learning is employed—coordination arises entirely from heuristic assignment and GA–DP optimisation.}
  \label{fig:Schematic}
\end{figure*}

Each agent’s state comprises its position $(x,y)$, heading $\theta$, and speed $v$, with actions defined by turn and acceleration commands constrained within fixed bounds. Agents have full visibility of all teammates and opponents. Blue team coordination emerges through a hierarchical heuristic framework: each agent executes a selected heuristic, while a hybrid genetic algorithm combined with dynamic programming optimises the heuristic assignments and parameters to maximise team-level performance.

\subsection{Overview}

Each Blue agent is controlled by a selected heuristic, which determines its low-level actions in the environment. Team-level coordination emerges through the assignment of heuristics to agents, rather than through explicit joint planning at every time step. To optimise these assignments and any associated heuristic parameters, we employ a hybrid genetic algorithm combined with dynamic programming, enabling efficient exploration of the combinatorial space of possible team configurations. The blue agents have no prior knowledge of attacker trajectories. Defensive strategies are optimised via the GA based on estimating strategies and evaluating them through interaction. In this way our approach is agnostic to the attacker model.

Figure \ref{fig:Schematic} illustrates the pipeline within the framework, showing how heuristic selection, environment interaction, and GA-DP optimisation are integrated to maximise the Blue team’s win rate. The simulation itself follows a standard Markov Decision Process (MDP) structure: at each time step, all agents submit actions, the environment transitions accordingly, and new observations are returned (along with potential rewards, though these are unused here). Unlike conventional reinforcement learning approaches, all decision-making in our framework is purely heuristic-driven, with coordination emerging from GA–DP optimisation rather than from policy-gradient or value-based learning.

This modular approach enables scaling by reusing and combining optimised heuristics without retraining.

\subsection{Hierarchical Policy as Chromosomes}

In our framework, the hierarchical policy for the Blue team is encoded as a chromosome suitable for GA optimisation. Each chromosome represents a full team configuration of $N$ agents, where each agent’s gene encodes its heuristic, spawn location, and heuristic-specific parameter. Formally, the chromosome $C$ can be written:

$$C=(g_1, g_2, \ldots, g_N), g_i = (H_i, W_i, P_i)$$

where $g_i$ is the gene of agent $i$, $H_i$ is the selected heuristic, $W_i$ is the spawn location, and $P_i$ is a heuristic-specific parameter. This representation allows the GA to explore combinations of agent strategies while maintaining the hierarchy of decision-making.

The hierarchical policy separates low-level behaviour, determined by the agent’s heuristic $H_i$, from high-level coordination, determined by the assignment of heuristics across the team $C$. During optimisation, crossover and mutation operate at the gene level, allowing recombination of heuristics and parameters between team configurations while preserving feasible individual behaviours. Dynamic programming complements the GA by efficiently evaluating the combinatorial space of heuristic assignments, enabling scalable optimisation of large teams. This chromosome-based representation ensures the search respects the hierarchy of decisions and allows the hybrid GA-DP to maximise team-level win rate effectively.

\subsection{Heuristics and Generation}

Each gene in a chromosome encodes a heuristic $H_i$ that governs the low-level behaviour of the corresponding Blue agent. The heuristics define action selection based on the current environment state, including the positions of Red agents, teammates, and the target. Heuristics are parameterised to allow flexibility in behaviour. By varying these parameters, a single heuristic template can generate multiple distinct strategies, increasing the diversity of candidate actions available for the GA-DP to optimise.

The use of heuristics is a deliberate design decision to avoid the complexity of skill discovery approaches such as feudal networks \cite{RN276}, variational auto-encoders \cite{RN275}, or general value functions \cite{RN274}, while still allowing explicit incorporation of action-masking and death-masking without additional training or the complexity of attention mechanisms. 

We design a set of modular heuristics that capture a range of tactical behaviours, such as intercepting attackers, maintaining defensive formations, and prioritising high-threat targets. Each heuristic encodes a distinct strategy, allowing the framework to combine them into larger teams while preserving interpretability and explicit action control. While one heuristic may be more familiar or conventional, focusing solely on it risks missing emergent cooperative behaviours. Accordingly, the heuristic set is curated for both diversity and effectiveness.

We utilize a combination of hand-crafted and LLM-generated heuristics. LLMs are employed to generate candidate heuristics due to their capacity to generate diverse strategies from concise textual prompts. The LLMs were provided a one-shot attempt at coding against our Python API, and made only minor coding errors. This allows us to substantially accelerate exploration of the heuristic design space and facilitate discovery of novel, and potentially effective, tactical behaviours with relatively less coding effort (only doing the JAX conversion manually) or other expert training. Unlike heuristics discovered through MARL, LLM-generated heuristics are immediately executable and generally interpretable. As all generated heuristics are assessed solely through quantitative simulation under identical conditions, any biases inherent in the language model are neutralised by the evolutionary selection process, ensuring only empirically effective strategies persist. A representative LLM prompt and the generated heuristic code is provided in the Supplementary Material.

\subsection{Stage 1 - GA for Chromosome Evolution}

In Stage 1 of our hybrid approach, a genetic algorithm (GA) is applied to small-scale team configurations—Red teams of 1–5 agents and Blue teams of 1–8 agents—to explore the space of hierarchical policies represented as chromosomes. Each chromosome encodes a full Blue team configuration, specifying the heuristic, spawn location, and parameterisation for each agent. Chromosome fitness is evaluated as the mean win rate over 128 simulated episodes.

The GA operates with a random initial population of 1024 chromosomes. Testing of Sobol~\cite{SOBOL196786} and Latin Hypercube~\cite{LHC} sampling population initialisation methods showed immaterial difference in performance. Evolution proceeds using 20\% elitism, tournament selection, and single-point crossover (at 70\%). Mutation introduces stochastic changes to genes to maintain population diversity. These operators allow the GA to efficiently explore a wide range of team-level strategies, exploiting high-performing chromosomes while exploring novel configurations.

Figure \ref{fig:baseHeatmaps} illustrates the GAs impact on team performance: heatmaps show win rates as a function of Blue and Red agent counts, comparing the random initial population with the GA-optimised population. The results demonstrate that GA evolution substantially improves team performance across small-scale configurations.

By restricting GA application to small-scale teams, we enable computationally feasible evaluation while producing a diverse set of high-quality chromosomes. These chromosomes serve as the candidate pool for Stage 2, where dynamic programming is applied to optimally assign heuristics and scale strategies to larger teams.

\begin{figure*}
\centering
\begin{subfigure}{.30\textwidth}
  \includegraphics[width=0.85\linewidth]{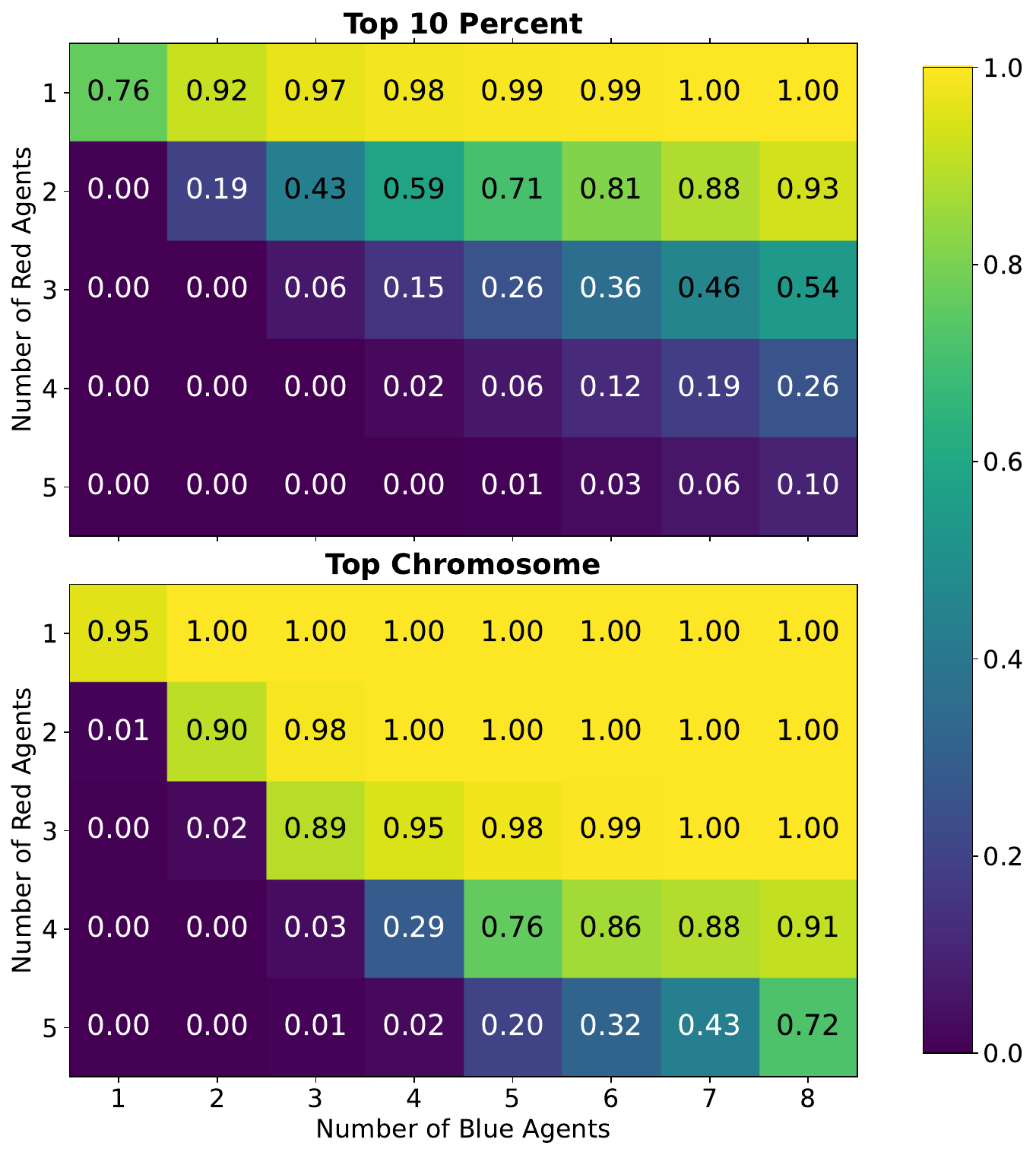}
  \caption{Heatmap without GA application.}
\end{subfigure}\hspace{3mm}
\begin{subfigure}{.30\textwidth}
  \includegraphics[width=0.85\linewidth]{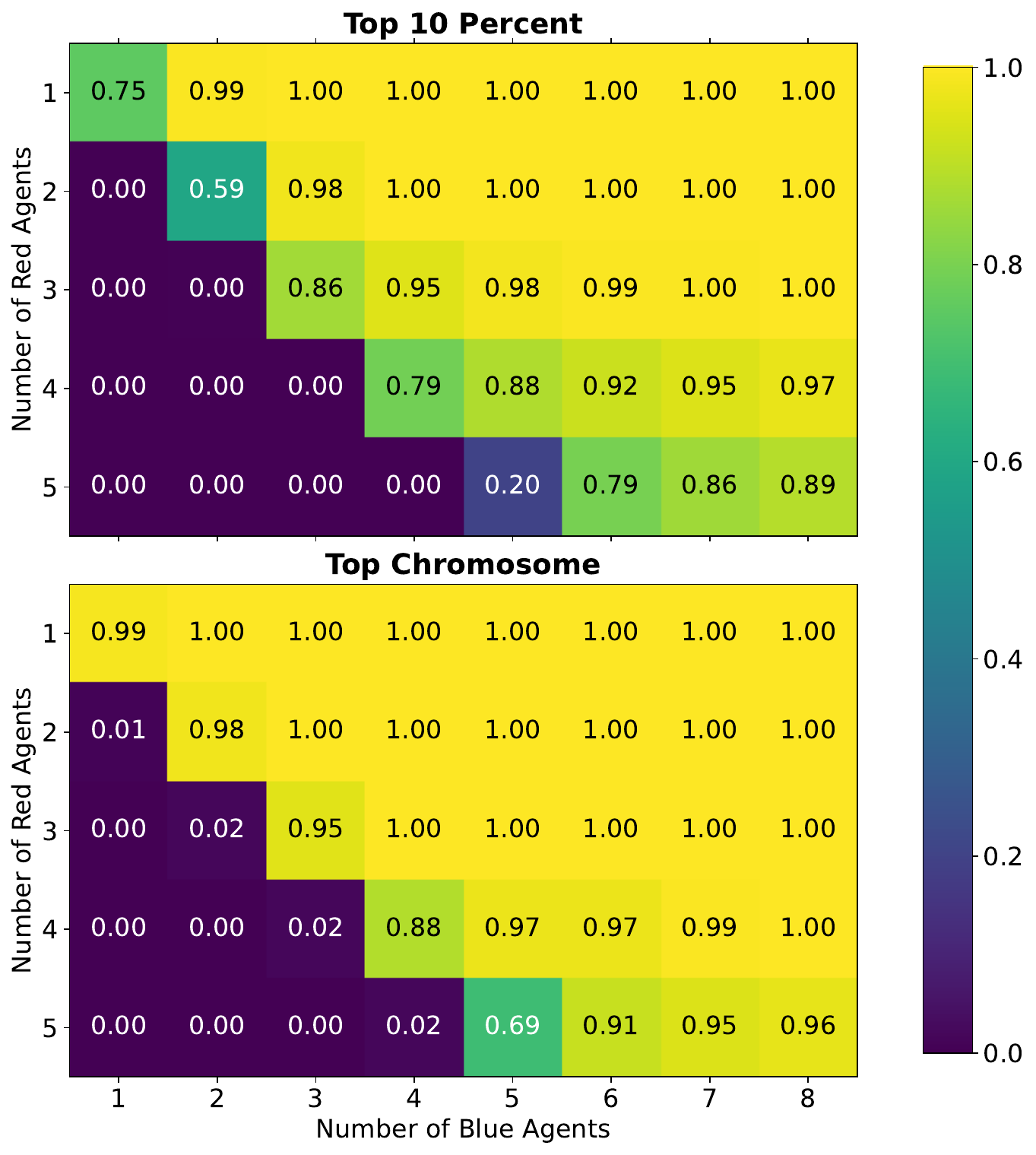}
  \caption{Heatmap following 128 GA generations.}
\end{subfigure}\hspace{3mm}
\caption{Heatmaps showing mean population win rates for the top-performing 10 percentile and the top-performing individual chromosome, both before (a) and after (b) application of GA generations. The latter, demonstrating significantly improved win rates for small scale swarm-defense scenarios, are utilised in subsequent stages of the approach.}
\label{fig:baseHeatmaps}
\end{figure*}

\subsection{Stage 2 - DP for Chromosome Allocation}

In the second stage, we address the core challenge of scaling defender policies from small scenarios to engagements with a larger number of attackers. Directly evolving effective chromosomes for scenarios with ten or more attackers is computationally prohibitive, as the search space grows combinatorially with the number of agents. Instead, we leverage the observation that outcomes of small-scale encounters, such as 2 vs. 3 or 3 vs. 4, can be represented compactly as empirical success probabilities obtained in Stage 1. These outcomes provide a lookup table of defender effectiveness against subgroups of attackers under varying team sizes.

Stage 2 uses a DP search to compose these small-scale results into effective strategies for larger engagements. The central idea is to represent a set of $R$ attackers as a partition into subgroups, for example $R=10 \rightarrow \{3,3,4\}$. Each subgroup is assumed to be engaged independently by a sub-team of defenders. For any given partition, the DP algorithm explores feasible allocations of the available defenders across the subgroups, consulting the Stage 1 heatmaps to evaluate the probability of success for each allocation.

Let $P(r,b)$ denote the probability of successfully defending against $r$ attackers with $b$ defenders, as estimated from Stage 1. For a partition $\pi=\{r_1,r_2,\ldots,r_k\}$ of $R$, and an allocation $\alpha=\{b_1,b_2,\ldots,b_k\}$ of $B$, the joint success probability is given by:

\begin{eqnarray}\label{eq:jointprob}
V(\pi,\alpha) = \prod_{i=1}^{k} P(r_i,b_i)
\end{eqnarray}

The DP search enumerates partitions of $R$, and for each partition computes the allocation $\alpha$ that maximises $V(\pi,\alpha)$. By caching intermediate results, the DP approach avoids redundant computation, enabling efficient evaluation of all possible subgroup structures and defender splits.

This stage therefore produces a set of high potential candidate defender configurations for larger-scale swarm-defense scenarios, which are then passed forward to Stage 3. In addition, this stage also enables exploration by trade-offs between success and resource expenditure by varying the number of defenders $B$, supporting systematic investigation of efficiency, for example through determining whether an 11-defender team achieves significantly higher success than a 10-defender team against the same set of attackers.

\begin{algorithm}[h]
\caption{Dynamic Program for Partitioning and Allocation}
\label{alg:dp-allocation}
\begin{algorithmic}
\Require \\$R$ (number of attackers) \\ $B$ (number of defenders) \\ $P(r,b)$ (Stage~1 success probabilities)
\Ensure Best partition $\pi^*$ and allocation $\alpha^*$
\Function{BestDefense}{$R,B$}
    \If{$R = 0$}
        \State \Return (success $=1$, allocation $=[]$)
    \EndIf
    \If{cache[$R,B$] exists}
        \State \Return cache[$R,B$]
    \EndIf
    \State $bestValue \gets 0$
    \State $bestSolution \gets \text{None}$
    \For{each partition $\pi = \{r_1, \ldots, r_k\}$ of $R$}
        \For{each allocation $\alpha = \{b_1, \ldots, b_k\}$ with $\sum b_i = B$}
            \State $value \gets \prod_{i=1}^{k} P(r_i, b_i)$
            \If{$value > bestValue$}
                \State $bestValue \gets value$
                \State $bestSolution \gets (\pi,\alpha)$
            \EndIf
        \EndFor
    \EndFor
    \State cache[$R,B$] $\gets$ (bestValue, bestSolution)
    \State \Return cache[$R,B$]
\EndFunction
\end{algorithmic}
\end{algorithm}

Consider the case $R=8$ attackers and $B=10$ defenders. The DP routine computes $f(8,10)$, the best defense probability for this scenario. One possible split is partition to $\{5,3\}$ attackers and allocate $\{7,3\}$ defenders. The value for this split is computed as:
$$f(8,10)=f(5,7) \times f(3,3)$$
Another split might be $\{4,3,1\}$ attackers with defender allocation $\{5,4,1\}$. By considering all such splits and retaining the maximum value, the algorithm builds $f(8,10)$ from small sub-problems. Each of these is computed recursively, terminating in base cases where outcomes are directly available from the Stage 1 heatmaps. Once a value is computed, it can be used in all higher-level states, thus avoiding redundant recomputation.

By way of contrast, a brute force approach would require enumerating all integer partitions of $R$, and for each partition considering all feasible allocations of defenders. The number of (unordered) partitions of an integer $R$ is given by the Bell number $B_R$ which approaches super-exponential growth. For a partition of size $k$, the number of possible allocations of defenders (including an empty allocation) is $\binom{B+k-1}{k-1}$. Combining these, it is becomes clear that overall complexity of brute force means it is intractable for all but small values of $R$. The DP approach on the other hand has polynomial complexity.

\begin{proposition}
\label{prop:dp-complexity}
Let $R$ and $B$ denote the number of red attackers and blue defenders, respectively, and let $k$ denote the maximum subgroup size of red agents for which precomputed outcomes are available. Then Algorithm~\ref{alg:dp-allocation} has worst-case time complexity
\[
O(R^3 B^3),
\]
which is polynomial in $R$ and $B$ for fixed $k$.
\end{proposition}

The Supplementary Material provides a proof for Proposition \ref{prop:dp-complexity}.

\subsection{Stage 3 - Chromosome Sampling with Priors}

The role of Stage 3 is to transition from small-scale empirical knowledge (Stage 1) and dynamic allocations (Stage 2) to full-sized configurations that can be tested in the target scenario. Rather than relying on evolutionary operators, Stage 3 systematically samples and combines high-performing small-scale chromosomes.

Stage 3 commences by sampling from the distribution of promising sub-chromosomes discovered in Stage 1, then assembling them into full chromosomes consistent with the allocations identified in Stage 2. For example, if the DP procedure recommends that a sub-team of three defenders be allocated to a partition of five attackers, Stage 3 samples multiple candidate sub-teams from the high-performing small-scale chromosomes identified in Stage 1. This sampling introduces diversity, recognising that a strong sub-team in isolation may not retain its effectiveness when combined with other sub-teams. Repeating this process across all partitions generates a population of candidate full-team chromosomes for evaluation in the full simulation.

These assembled chromosomes are then evaluated directly in the full simulation environment, where the stochastic dynamics and interdependencies between sub-teams are faithfully represented. This step is crucial, as the independence assumptions underpinning Stage 2’s DP analysis can only provide an approximation of real-world multi-agent interactions.

By systematically joining and testing factorised chromosome components, Stage 3 achieves two goals: (i) it leverages prior knowledge to drastically reduce the search space compared to naïve sampling, and (ii) it ensures that the retained strategies are robust to the complexities of large-scale multi-agent coordination.

\subsection{Stage 4 - Iterative Refinement}

During Stage 4, win rate performance of the full-team chromosomes are attributed back to the sub-chromosomes that were utilised. For each sub-chromosome, we accumulate a record of its performance across all simulations in which it appeared (it may be utilised in several full-team chromosomes), enabling the calculation of a mean win rate that reflects its effectiveness in different contexts. This process establishes an empirical performance profile for each sub-chromosome, which serves as the basis for the weighted refinement applied in Stage 4. By grounding sub-chromosome estimates in repeated observations, Stage 3 provides the link between small-scale evaluations and the more comprehensive large-scale simulations.

For each sub-chromosome utilised in the large-scale evaluations, its win rate estimate is refined by combining priors from Stage 1 with empirical evidence accumulated in Stage 3. Specifically, the refined estimate, $w^{(1)}_\textit{new}$, is obtained through a weighted combination:

$$w^{(1)}_\textit{new}=\frac{K}{n+K} w^{(1)}_\textit{old} + \frac{n}{n+K} w^{(3)}$$

where $w^{(1)}_\textit{old}$ is the current Stage 1 estimate, $w^{(3)}$ is the average performance observed across Stage 3 simulations, and $n$ is the number of times the sub-chromosome has featured in those simulations. Here $K$ is a hyper-parameter that controls the relative influence of the prior versus empirical data.

The value of $K$ was tuned via sensitivity analysis and trial-and-error, balancing the need for meaningful updates against the risk of overreacting to noisy outcomes. A value of $K = 200$ was ultimately selected, providing smooth yet efficient refinement of the Stage 1 estimates.

This weighted refinement ensures that sub-chromosomes with consistent large-scale success are promoted, while those that under-perform are down-weighted, even if they appeared promising in small-scale settings. By iteratively applying this procedure across all evaluated chromosomes, Stage 4 produces more reliable win rate estimates for each building block. As evaluations accumulate, the refined win rates converge toward values that accurately reflect performance under large-scale conditions. Stage 4 thus validates composite strategies at scale while iteratively sharpening the estimated effectiveness of the underlying building blocks.

\section{Results and discussion}

\subsection{Baseline}

Performance in all cases is measured as the defender win rate, defined as the proportion of simulations in which defenders successfully prevent attackers from reaching the target. Since prior work in multi-agent attack–defense scenarios has focused almost exclusively on small-scale settings (typically involving fewer than ten agents) or algebraic formulations that do not scale, no directly comparable baselines exist in the literature.  Accordingly, we construct two internal baselines to anchor our large-scale evaluation.

\textit{Random population}. In this baseline, defender chromosomes are generated by assigning heuristics and parameters  at random. This provides a lower-bound performance level and establishes how well defenders can perform without any optimisation.

\textit{Genetic algorithm}. In this baseline, we apply the same GA used to optimise the small-scale scenarios in Stage 1 directly to large-scale scenarios, starting from a random initial population. This represents the natural alternative to our staged approach, in which a conventional GA is tasked with discovering effective defender strategies without factorisation or staged refinement.

Together, these baselines capture the extremes of unstructured search: a purely random allocation and a flat GA applied directly to the full problem. By comparing against these, we can demonstrate the benefits of staged factorisation, DP, and large-scale refinement in producing scalable and effective defender strategies.

As the GA baseline only exhibits minor performance enhancement over the random baseline, we show only the former in Figure \ref{fig:largeResultsA}. Overall win rates remain close to zero across most scenarios, and even the best-performing chromosomes achieve only modest outcomes (e.g., a maximum win rate of 0.328 for 10 Red at a 1.5 Blue/Red ratio). Performance collapses almost entirely for larger attacker swarms, with even elite win rates rarely exceeding 0.1.

These results clearly establish the difficulty of the large-scale defense problem, setting the stage for Section 4.2, where we demonstrate our hybrid GA–DP approach substantially outperforms baselines, achieving high win rates even against dense attacker swarms.

\subsection{Comparative Results}

We perform an ablation study to quantify the contribution of each stage in our hybrid GA–DP pipeline. Figure \ref{fig:largeResults} present heatmaps of defender win rates for three configurations: (a) GA baseline, (b) Stage 3: DP-informed chromosome sampling, and (c) Stage 4: Iterative Refinement. This allows a clear evaluation of the incremental benefit of each component.

\begin{figure*}[htb!]
\centering
\begin{subfigure}{.27\textwidth}
  \includegraphics[width=\linewidth]{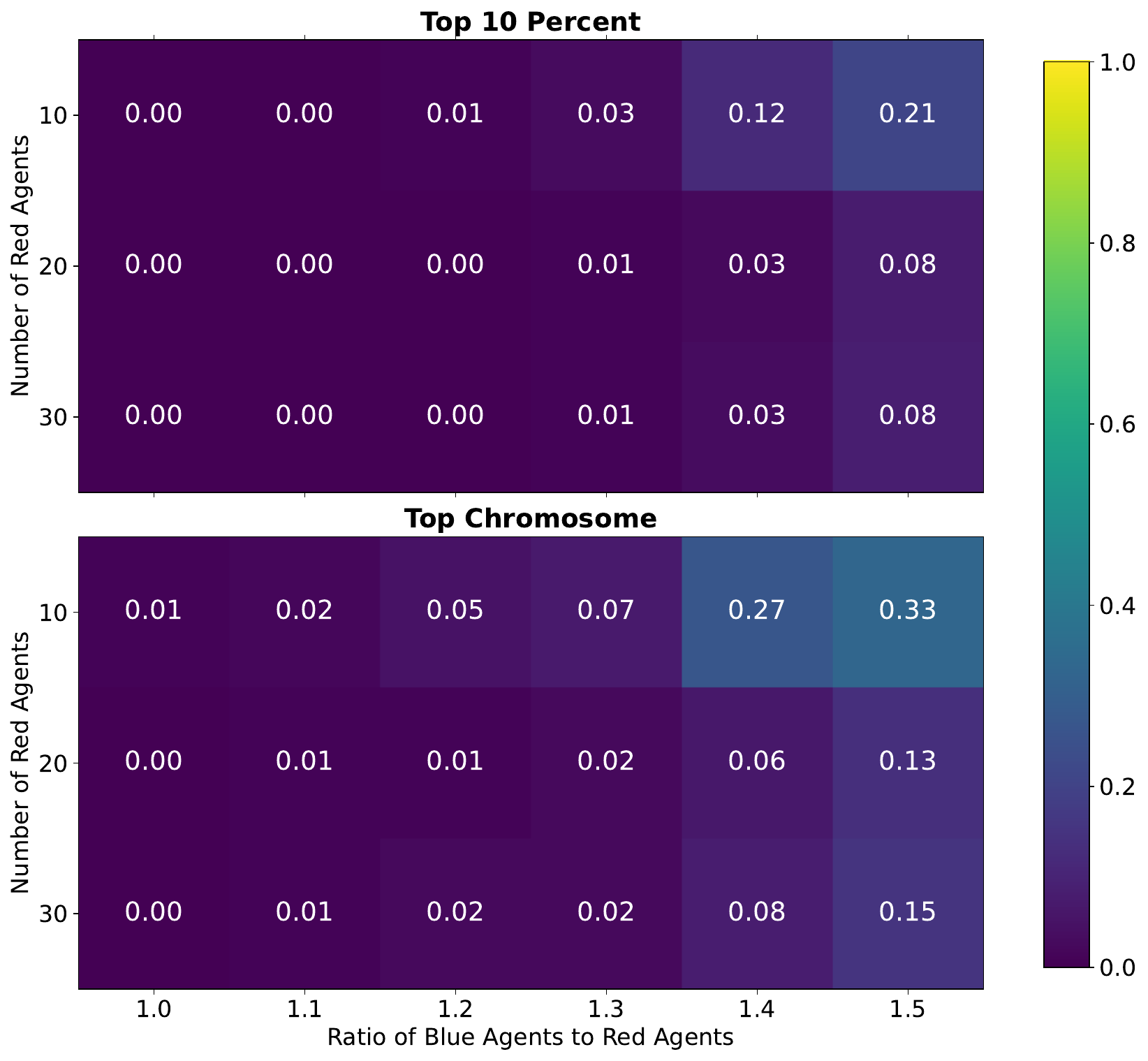}
  \caption{GA Baseline.}
  \label{fig:largeResultsA}
\end{subfigure}\hspace{3mm}
\begin{subfigure}{.27\textwidth}
  \includegraphics[width=\linewidth]{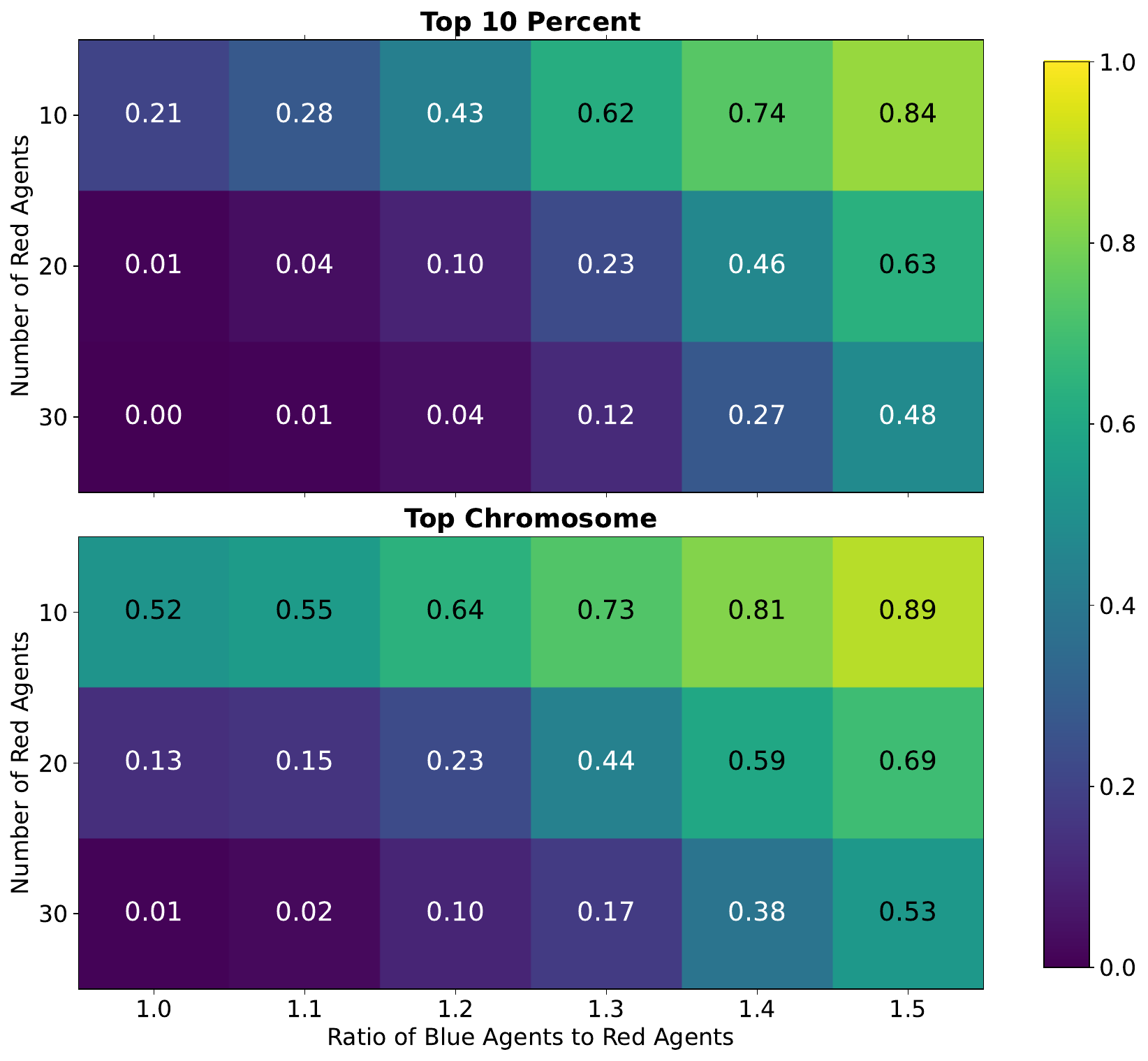}
  \caption{Stage 3: DP-informed Sampling.}
  \label{fig:largeResultsB}
\end{subfigure}\hspace{3mm}
\begin{subfigure}{.27\textwidth}
  \includegraphics[width=\linewidth]{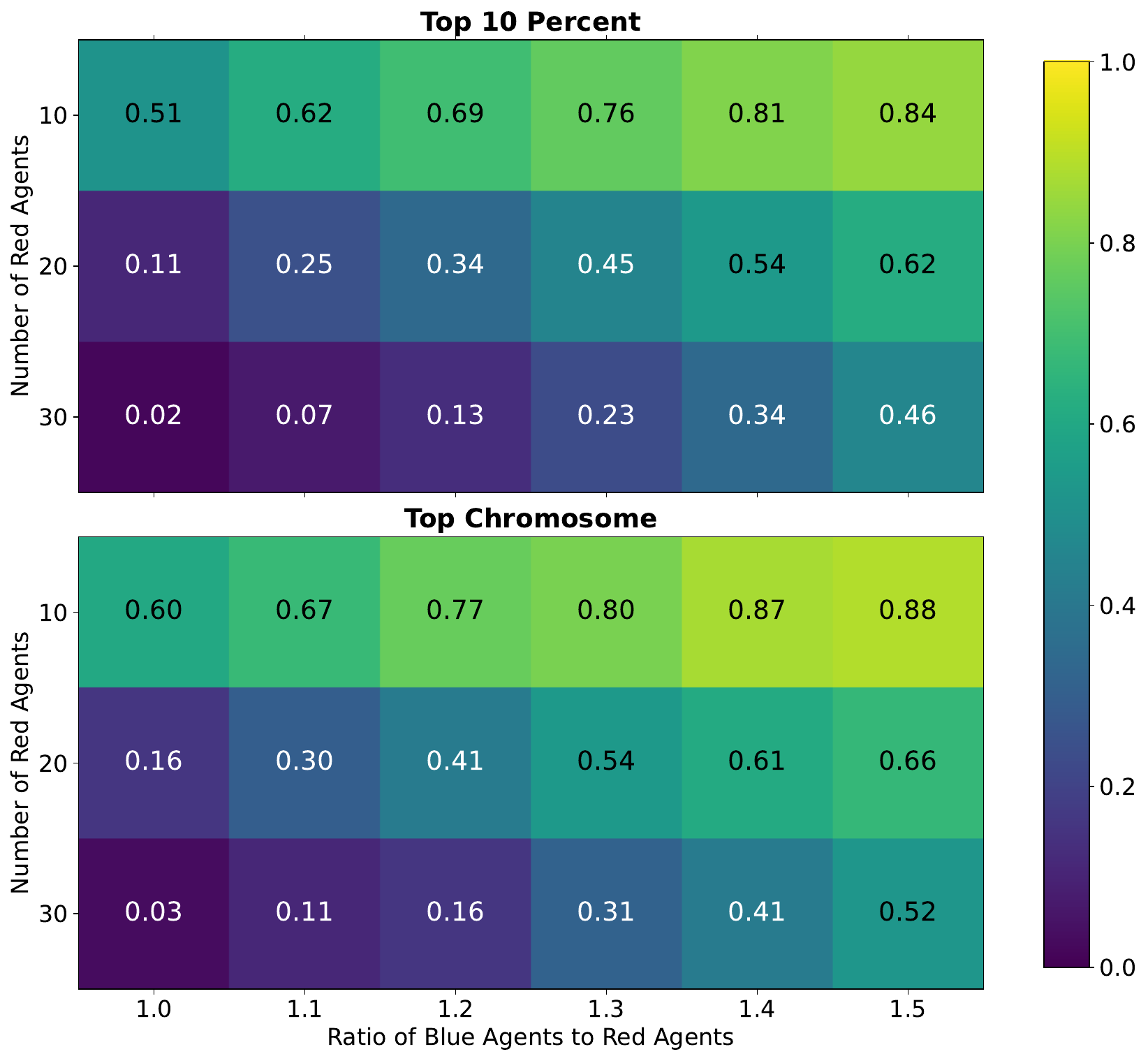}
  \caption{Stage 4: Iterative Refinement.}
  \label{fig:largeResultsC}
\end{subfigure}\hspace{3mm}
\caption{Ablation study of defender performance across Red swarm sizes and Blue-to-Red ratios. Columns show Blue/Red ratios and rows show the number of Red attackers. Shading indicates win rate, with lighter colors representing higher performance. For each subfigure, the top row shows the mean win rate of the top 10\% of chromosomes and the bottom row the best-performing chromosome. Stage 3 demonstrates the benefit of combining DP allocation with sampling, while Stage 4 illustrates the additional gains from iterative refinement, achieving the highest and most consistent win rates across all scenarios.}
\label{fig:largeResults}
\end{figure*}

\textit{Stage 3 Ablation (DP-informed Chromosome Sampling)}.
Stage 3 (Figure \ref{fig:largeResultsB}) combines the dynamic programming allocation from Stage 2 with sampling of promising small-scale sub-chromosomes to assemble candidate full-team strategies. This stage alone substantially improves performance relative to the GA baseline (Figure \ref{fig:largeResultsA}), producing nontrivial defender strategies across swarm sizes. For example, at $R=20$ and a 1.5 Blue/Red ratio, the best chromosome achieves a win rate of 0.69 compared to just 0.13 in the baseline. However, Stage 3 alone struggles to optimise: while performance improves, win rates for some smaller Blue/Red agent ratios remain modest (e.g., 0.17 at $R=30$,$B/R=1.3$). These results highlight that DP-informed sampling establishes a strong foundation but is insufficient for fully robust large-scale strategies.

\textit{Stage 4 Ablation (Iterative Refinement)}.
Incorporating Stage 4 refinement (Figure \ref{fig:largeResultsC}) significantly amplifies performance across all scenarios. Iterative re-weighting and evaluation of sub-chromosomes consistently produces chromosomes that outperform Stage 3, particularly at lower Blue/Red ratios where coordination is critical. For instance, at $R=20$ and $B/R=1.3$, the top 10 percent improves from 0.22 (Stage 3) to 0.45 after Stage 4 refinement. For smaller swarms, defenders approach near-complete reliability, with win rates exceeding 0.80 at $R=10$ and $B/R=1.3$. Even the most challenging scenarios with $R=30$ achieve a top win rate of 0.52, far above the negligible performance of the baseline.

Overall, this ablation study demonstrates that Stage 3 establishes a foundation of promising strategies, but Stage 4 is essential for convergence toward robust, high-performing defender allocations. This refinement allows the framework to scale efficiently, preserving effectiveness even as the size and complexity of the attacking swarm increase.

\subsection{Discussion}

Scalable swarm defense poses an inherently hard optimisation problem due to the explosive combinatorial growth of possible defender allocations. The staged approach presented here substantially outperforms the baseline (random initialisation + GA) and achieves robust performance across increasingly challenging scenarios.

\textit{Breaking complexity into solvable pieces}.
The central advantage of the staged pipeline lies in factorisation of the search space. By first evolving small chromosomes in Stage~1, we obtain candidate building blocks under low-dimensional constraints. Stage~3 then leverages these factors through dynamic programming, assembling them into large-scale defender allocations. Finally, Stage~4 corrects biases introduced in earlier stages by iteratively refining chromosomes under direct evaluation in large scenarios. This decomposition dramatically reduces the effective search complexity compared to tackling large-scale swarm defense directly.

\textit{Refinement turns small-scale wins into large-scale victories}.
The comparison of Stages~3 and~4 highlights that some heuristics are strongly context-dependent. Certain heuristics perform well in small battles but fail to generalise when scaled, while others remain robust across sizes. Pruning experiments show that removing low-frequency heuristics has little effect, whereas removing the most frequent, high-impact heuristic substantially reduces performance, indicating that robustness depends on preserving key strategies. Iterative refinement in Stage~4 exposes these differences and converges on strategies that persist under large swarm conditions. This underscores the importance of multi-scale evaluation: assessing only at small scales risks over-committing to heuristics that collapse when scaled up.

\textit{Synergy, not parts: why chromosomes matter}.
An additional experiment explored gene-level sampling in Stage~3, where individual tuples were recombined independently rather than preserving entire chromosomes. Results were markedly poor. High-performing strategies appear to depend on synergistic combinations of tuples rather than on the marginal quality of individual components. Tuples that appear promising in isolation may be harmful if overused, while seemingly weak tuples can be critical in the right context. This resonates with the idea of deceptive landscapes \cite{Whitley1991, GoldbergDeceptive} and with findings in modular evolution \cite{RN461}, where interactions between submodules complicate search. In our setting, the effective unit of selection is the chromosome: the landscape is factorisable at the chromosome level but becomes brittle at the gene level.

\begin{figure}[htbp!]
\centering
\begin{subfigure}{.15\textwidth}
  \includegraphics[width=0.8\linewidth]{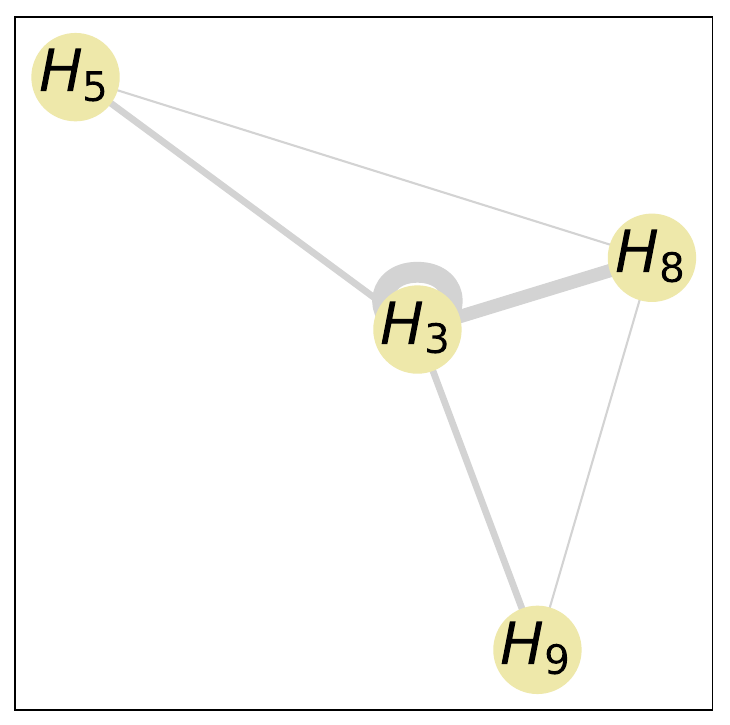}
  \caption{4R v. 7B \\building block.}
  \label{fig:networkxA}
\end{subfigure}\hspace{1mm}
\begin{subfigure}{.15\textwidth}
  \includegraphics[width=0.8\linewidth]{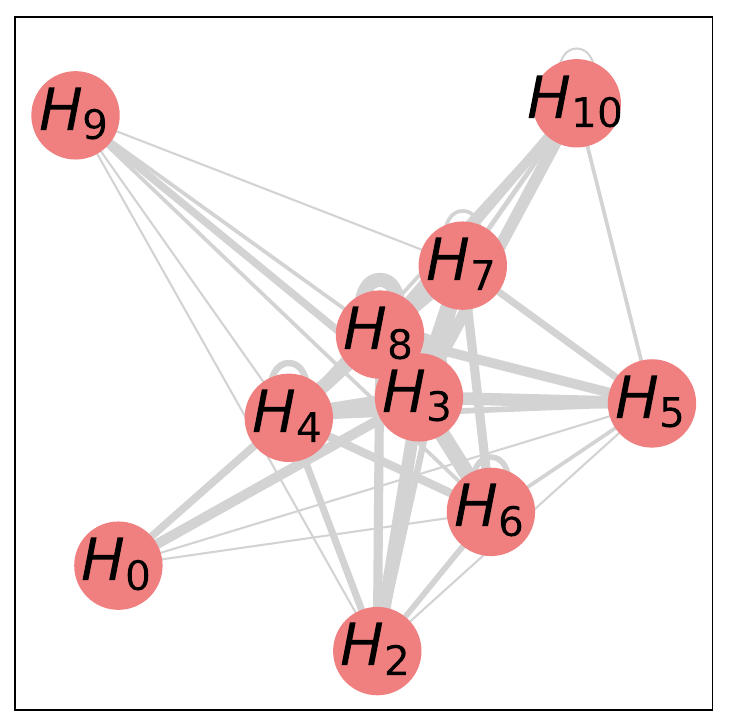}
  \caption{10R v. 15B \\Stage~1 only.}
  \label{fig:networkxB}
\end{subfigure}\hspace{1mm}
\begin{subfigure}{.15\textwidth}
  \includegraphics[width=0.8\linewidth]{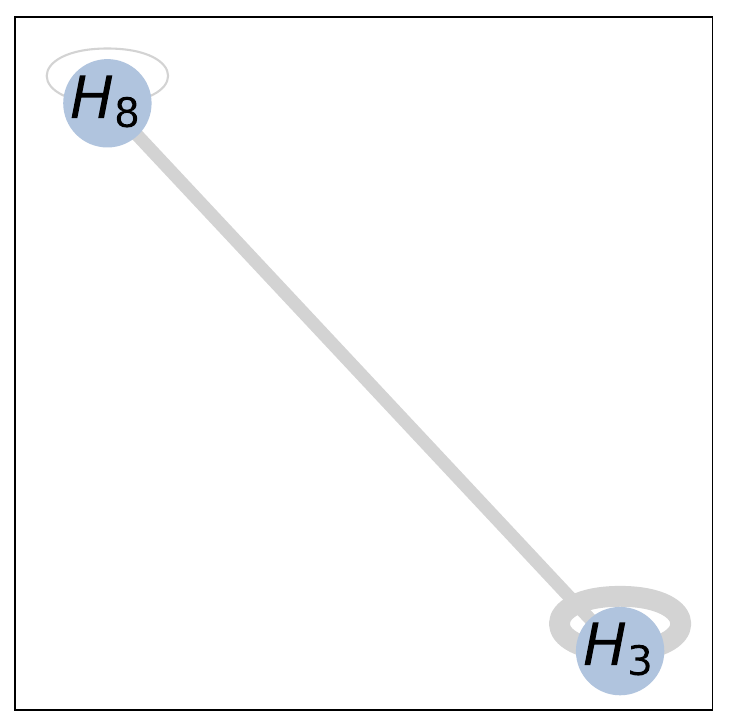}
  \caption{10R v. 15B \\iterative Stages~1--4.}
  \label{fig:networkxC}
\end{subfigure}\hspace{1mm}
\caption{Network graphs of heuristic co-occurrence (edge thickness proportional to co-occurrence frequency). GA alone can distil effective structures in small-scale cases (a), but fails to simplify when scaled up (b). Only the full pipeline (c) restores parsimony, converging on a compact set of synergistic heuristics.}
\label{fig:networkx}
\end{figure}

Figure~\ref{fig:networkx} illustrates this effect. Each network graph shows gene co-occurrence within the top 10\% of high-performing chromosomes in the final population (simplified to heuristics only for clarity), with edge weights proportional to co-occurrence frequency. Panel~(a) shows that GA alone can distil effective combinations in a small setting (4R v.~7B). In contrast, panel~(b) shows that the same process fails in the larger 10R v.~15B case, where the search space overwhelms direct GA optimisation. The full staged pipeline in panel~(c) restores structure, producing simplified chromosomes dominated by heuristics \#3 Predictive Interception and \#8 Threat-Level Assessment, sometimes supported by \#2 Circle Target or \#4 Rotation Patrol in low-ratio settings. These patterns show how staged evolution isolates a compact set of repeatable heuristics, balancing proactive interception with minimal passive guarding, and sustaining effectiveness even in complex scenarios.

In summary, the staged pipeline succeeds because it balances tractable search with large-scale fidelity. Early stages discover and preserve useful strategic structures, while later stages refine them under the full complexity of large engagements. Reducing the problem further to the gene level destroy this synergy, confirming that in this MAS defense setting the meaningful building blocks are entire chromosomes, not individual genes. These findings highlight both the limitations of naive factorisation and the promise of staged, chromosome-level approaches for large-scale multi-agent defense.

\section{Conclusions}

This study introduces a staged factorization approach for optimizing defender allocation in large-scale multi-agent drone defense scenarios. The key innovation lies in decomposing the problem into manageable stages: small-scale genetic evaluation of heuristics, factor refinement, and large-scale evaluation of combined strategies. This staged approach reduces search complexity, mitigates early evaluation biases, and enables the identification of high-performing defender allocations that generalize across swarm sizes — outcomes that conventional methods struggle to achieve.\balance

Our results highlight that some heuristics are highly effective only in small battles, whereas others provide robust defense in large-scale engagements. The method’s ability to combine these heuristics into effective large-scale strategies demonstrates the practical value of staged factorization in multi-agent systems. In real-world drone defense contexts, this implies that a principled, factorized evaluation can guide the deployment of autonomous agents to protect high-value assets against diverse adversary strategies.

An additional insight is the explicit mapping between defensive effectiveness and available resources. By evaluating defender allocations across a range of blue-to-red agent ratios, the approach reveals diminishing returns for certain heuristics and identifies allocation strategies that achieve high win rates with minimal resources. This provides actionable guidance for resource-constrained operations, allowing decision-makers to balance effectiveness against the number of deployed agents and optimize overall system efficiency.

A limitation of our approach is the assumption of full observability and deterministic agent dynamics. Real-world environments may feature partial information, communication constraints, and stochastic disturbances, which could affect generalizability of our findings. Additionally, the current heuristics and allocation schemes do not incorporate learning-based adaptation during deployment.

Future work could address these limitations by integrating partially observable decision-making, adaptive heuristics, and reinforcement learning at lower levels. Expanding to dynamic adversary behaviors, environmental obstacles, and more heterogeneous agent capabilities could further enhance the applicability of the approach to operational drone defense scenarios.


\section*{Acknowledgements}
 G. Douglas has received support through the Australian Government Research Training Program Scholarship. This work has also been supported in part by Saab Australia Pty Ltd through the Defence Trailblazer program. This work was supported with super-computing resources provided by the Phoenix HPC service at Adelaide University.


\clearpage
\bibliographystyle{ACM-Reference-Format} 
\bibliography{references}


\begin{thebibliography}{30}


\ifx \showCODEN    \undefined \def \showCODEN     #1{\unskip}     \fi
\ifx \showDOI      \undefined \def \showDOI       #1{#1}\fi
\ifx \showISBNx    \undefined \def \showISBNx     #1{\unskip}     \fi
\ifx \showISBNxiii \undefined \def \showISBNxiii  #1{\unskip}     \fi
\ifx \showISSN     \undefined \def \showISSN      #1{\unskip}     \fi
\ifx \showLCCN     \undefined \def \showLCCN      #1{\unskip}     \fi
\ifx \shownote     \undefined \def \shownote      #1{#1}          \fi
\ifx \showarticletitle \undefined \def \showarticletitle #1{#1}   \fi
\ifx \showURL      \undefined \def \showURL       {\relax}        \fi
\providecommand\bibfield[2]{#2}
\providecommand\bibinfo[2]{#2}
\providecommand\natexlab[1]{#1}
\providecommand\showeprint[2][]{arXiv:#2}

\bibitem[\protect\citeauthoryear{Achiam, Edwards, Amodei, and Abbeel}{Achiam et~al\mbox{.}}{2018}]%
        {RN275}
\bibfield{author}{\bibinfo{person}{Joshua Achiam}, \bibinfo{person}{Harrison Edwards}, \bibinfo{person}{Dario Amodei}, {and} \bibinfo{person}{Pieter Abbeel}.} \bibinfo{year}{2018}\natexlab{}.
\newblock \showarticletitle{Variational option discovery algorithms}.
\newblock \bibinfo{journal}{\emph{arXiv preprint arXiv:1807.10299}}  \bibinfo{volume}{abs/1807.10299} (\bibinfo{year}{2018}).
\newblock


\bibitem[\protect\citeauthoryear{Ahner and Parson}{Ahner and Parson}{2013}]%
        {Ahner2013}
\bibfield{author}{\bibinfo{person}{Darryl Ahner} {and} \bibinfo{person}{Carl Parson}.} \bibinfo{year}{2013}\natexlab{}.
\newblock \showarticletitle{Weapon tradeoff analysis using dynamic programming for a dynamic weapon target assignment problem within a simulation}. In \bibinfo{booktitle}{\emph{2013 Winter Simulations Conference (WSC)}}. \bibinfo{publisher}{IEEE Press}, \bibinfo{address}{Washington, D.C.}, \bibinfo{pages}{2831--2841}.
\newblock
\urldef\tempurl%
\url{https://doi.org/10.1109/WSC.2013.6721653}
\showDOI{\tempurl}


\bibitem[\protect\citeauthoryear{Bradbury, Frostig, Hawkins, Johnson, Leary, Maclaurin, Necula, Paszke, Vander{P}las, Wanderman-{M}ilne, and Zhang}{Bradbury et~al\mbox{.}}{2018}]%
        {jax}
\bibfield{author}{\bibinfo{person}{James Bradbury}, \bibinfo{person}{Roy Frostig}, \bibinfo{person}{Peter Hawkins}, \bibinfo{person}{Matthew~James Johnson}, \bibinfo{person}{Chris Leary}, \bibinfo{person}{Dougal Maclaurin}, \bibinfo{person}{George Necula}, \bibinfo{person}{Adam Paszke}, \bibinfo{person}{Jake Vander{P}las}, \bibinfo{person}{Skye Wanderman-{M}ilne}, {and} \bibinfo{person}{Qiao Zhang}.} \bibinfo{year}{2018}\natexlab{}.
\newblock \bibinfo{booktitle}{\emph{{JAX}: composable transformations of {P}ython+{N}um{P}y programs}}.
\newblock Google Research.
\newblock
\urldef\tempurl%
\url{http://github.com/jax-ml/jax}
\showURL{%
\tempurl}


\bibitem[\protect\citeauthoryear{DeepMind}{DeepMind}{2025}]%
        {gemini}
\bibfield{author}{\bibinfo{person}{Google DeepMind}.} \bibinfo{year}{2025}\natexlab{}.
\newblock \bibinfo{title}{Gemini}.
\newblock
\newblock
\urldef\tempurl%
\url{https://www.deepmind.com/gemini}
\showURL{%
\tempurl}


\bibitem[\protect\citeauthoryear{Eysenbach, Gupta, Ibarz, and Levine}{Eysenbach et~al\mbox{.}}{2018}]%
        {RN274}
\bibfield{author}{\bibinfo{person}{Benjamin Eysenbach}, \bibinfo{person}{Abhishek Gupta}, \bibinfo{person}{Julian Ibarz}, {and} \bibinfo{person}{Sergey Levine}.} \bibinfo{year}{2018}\natexlab{}.
\newblock \showarticletitle{Diversity is all you need: Learning skills without a reward function}.
\newblock \bibinfo{journal}{\emph{arXiv preprint arXiv:1802.06070}}  \bibinfo{volume}{abs/1802.06070} (\bibinfo{year}{2018}).
\newblock


\bibitem[\protect\citeauthoryear{Ghavamzadeh, Mahadevan, and Makar}{Ghavamzadeh et~al\mbox{.}}{2006}]%
        {RN247}
\bibfield{author}{\bibinfo{person}{Mohammad Ghavamzadeh}, \bibinfo{person}{Sridhar Mahadevan}, {and} \bibinfo{person}{Rajbala Makar}.} \bibinfo{year}{2006}\natexlab{}.
\newblock \showarticletitle{Hierarchical multi-agent reinforcement learning}.
\newblock \bibinfo{journal}{\emph{Autonomous Agents and Multi-Agent Systems}} \bibinfo{volume}{13}, \bibinfo{number}{2} (\bibinfo{year}{2006}), \bibinfo{pages}{197--229}.
\newblock
\showISSN{1387-2532}
\urldef\tempurl%
\url{https://doi.org/10.1007/s10458-006-7035-4}
\showDOI{\tempurl}


\bibitem[\protect\citeauthoryear{Goldberg, Deb, and Horn}{Goldberg et~al\mbox{.}}{1994}]%
        {GoldbergDeceptive}
\bibfield{author}{\bibinfo{person}{David Goldberg}, \bibinfo{person}{Kalyanmoy Deb}, {and} \bibinfo{person}{Jeffrey Horn}.} \bibinfo{year}{1994}\natexlab{}.
\newblock \showarticletitle{Massive Multimodality, Deception, and Genetic Algorithms}.
\newblock \bibinfo{journal}{\emph{Parallel Problem Solving from Nature}}  \bibinfo{volume}{2} (\bibinfo{date}{08} \bibinfo{year}{1994}).
\newblock


\bibitem[\protect\citeauthoryear{Goldberg, Korb, and Deb}{Goldberg et~al\mbox{.}}{1989}]%
        {goldberg1989messy}
\bibfield{author}{\bibinfo{person}{David~E Goldberg}, \bibinfo{person}{Bradley Korb}, {and} \bibinfo{person}{Kalyanmoy Deb}.} \bibinfo{year}{1989}\natexlab{}.
\newblock \showarticletitle{Messy genetic algorithms: Motivation, analysis, and first results}.
\newblock \bibinfo{journal}{\emph{Complex systems}} \bibinfo{volume}{3}, \bibinfo{number}{5} (\bibinfo{year}{1989}), \bibinfo{pages}{493--530}.
\newblock


\bibitem[\protect\citeauthoryear{Huaiping, Jingxu, Yingwu, and Hao}{Huaiping et~al\mbox{.}}{2006}]%
        {Cai2006}
\bibfield{author}{\bibinfo{person}{Cai Huaiping}, \bibinfo{person}{Liu Jingxu}, \bibinfo{person}{Chen Yingwu}, {and} \bibinfo{person}{Wang Hao}.} \bibinfo{year}{2006}\natexlab{}.
\newblock \showarticletitle{Survey of the research on dynamic weapon-target assignment problem}.
\newblock \bibinfo{journal}{\emph{Journal of Systems Engineering and Electronics}} \bibinfo{volume}{17}, \bibinfo{number}{3} (\bibinfo{year}{2006}), \bibinfo{pages}{559--565}.
\newblock
\urldef\tempurl%
\url{https://doi.org/10.1016/S1004-4132(06)60097-2}
\showDOI{\tempurl}


\bibitem[\protect\citeauthoryear{Isaacs}{Isaacs}{1951}]%
        {RN177}
\bibfield{author}{\bibinfo{person}{Rufus Isaacs}.} \bibinfo{year}{1951}\natexlab{}.
\newblock \bibinfo{booktitle}{\emph{Games of Pursuit}}.
\newblock \bibinfo{publisher}{RAND Corporation}, \bibinfo{address}{Santa Monica, CA}.
\newblock


\bibitem[\protect\citeauthoryear{Isaacs}{Isaacs}{1965}]%
        {RN206}
\bibfield{author}{\bibinfo{person}{Rufus Isaacs}.} \bibinfo{year}{1965}\natexlab{}.
\newblock \bibinfo{booktitle}{\emph{Differential games : a mathematical theory with applications to warfare and pursuit, control and optimization}}.
\newblock \bibinfo{publisher}{Wiley}, \bibinfo{address}{New York}.
\newblock
\newblock
\shownote{Bibliography: p. 379-380.}


\bibitem[\protect\citeauthoryear{Kline, Ahner, and Hill}{Kline et~al\mbox{.}}{2019}]%
        {KLINE2019226}
\bibfield{author}{\bibinfo{person}{Alexander Kline}, \bibinfo{person}{Darryl Ahner}, {and} \bibinfo{person}{Raymond Hill}.} \bibinfo{year}{2019}\natexlab{}.
\newblock \showarticletitle{The Weapon-Target Assignment Problem}.
\newblock \bibinfo{journal}{\emph{Computers \& Operations Research}}  \bibinfo{volume}{105} (\bibinfo{year}{2019}), \bibinfo{pages}{226--236}.
\newblock
\showISSN{0305-0548}
\urldef\tempurl%
\url{https://doi.org/10.1016/j.cor.2018.10.015}
\showDOI{\tempurl}


\bibitem[\protect\citeauthoryear{Kulkarni, Narasimhan, Saeedi, and Tenenbaum}{Kulkarni et~al\mbox{.}}{2016}]%
        {RN273}
\bibfield{author}{\bibinfo{person}{Tejas~D. Kulkarni}, \bibinfo{person}{Karthik~R. Narasimhan}, \bibinfo{person}{Ardavan Saeedi}, {and} \bibinfo{person}{Joshua~B. Tenenbaum}.} \bibinfo{year}{2016}\natexlab{}.
\newblock \showarticletitle{Hierarchical deep reinforcement learning: integrating temporal abstraction and intrinsic motivation}. In \bibinfo{booktitle}{\emph{Proceedings of the 30th International Conference on Neural Information Processing Systems}} (Barcelona, Spain) \emph{(\bibinfo{series}{NIPS'16})}. \bibinfo{publisher}{Curran Associates Inc.}, \bibinfo{address}{Red Hook, NY, USA}, \bibinfo{pages}{3682–3690}.
\newblock
\showISBNx{9781510838819}


\bibitem[\protect\citeauthoryear{Louette, Leroy, Geurts, and Ernst}{Louette et~al\mbox{.}}{2025}]%
        {RN386}
\bibfield{author}{\bibinfo{person}{Arthur Louette}, \bibinfo{person}{Pascal Leroy}, \bibinfo{person}{Yanis Geurts}, {and} \bibinfo{person}{Damien Ernst}.} \bibinfo{year}{2025}\natexlab{}.
\newblock \bibinfo{title}{Existing Gaps In Reinforcement Learning For Drone Warfare}.  (\bibinfo{year}{2025}).
\newblock
\newblock
\shownote{Author preprint.}


\bibitem[\protect\citeauthoryear{Mckay, Beckman, and Conover}{Mckay et~al\mbox{.}}{1979}]%
        {LHC}
\bibfield{author}{\bibinfo{person}{Michael Mckay}, \bibinfo{person}{Richard Beckman}, {and} \bibinfo{person}{William Conover}.} \bibinfo{year}{1979}\natexlab{}.
\newblock \showarticletitle{A Comparison of Three Methods for Selecting Vales of Input Variables in the Analysis of Output From a Computer Code}.
\newblock \bibinfo{journal}{\emph{Technometrics}}  \bibinfo{volume}{21} (\bibinfo{date}{05} \bibinfo{year}{1979}), \bibinfo{pages}{239--245}.
\newblock
\urldef\tempurl%
\url{https://doi.org/10.1080/00401706.1979.10489755}
\showDOI{\tempurl}


\bibitem[\protect\citeauthoryear{Nathanson}{Nathanson}{2000}]%
        {Nathanson2000}
\bibfield{author}{\bibinfo{person}{Melvyn~B. Nathanson}.} \bibinfo{year}{2000}\natexlab{}.
\newblock \bibinfo{booktitle}{\emph{Elementary Methods in Number Theory} (\bibinfo{edition}{1st ed. 2000.} ed.)}.
\newblock \bibinfo{publisher}{Springer New York}, \bibinfo{address}{New York, NY}.
\newblock
\showISBNx{1-280-01072-X}


\bibitem[\protect\citeauthoryear{Ning and Xie}{Ning and Xie}{2024}]%
        {RN80}
\bibfield{author}{\bibinfo{person}{Zepeng Ning} {and} \bibinfo{person}{Lihua Xie}.} \bibinfo{year}{2024}\natexlab{}.
\newblock \showarticletitle{A survey on multi-agent reinforcement learning and its application}.
\newblock \bibinfo{journal}{\emph{Journal of Automation and Intelligence}} \bibinfo{volume}{3}, \bibinfo{number}{2} (\bibinfo{year}{2024}), \bibinfo{pages}{73--91}.
\newblock
\showISSN{2949-8554}
\urldef\tempurl%
\url{https://doi.org/10.1016/j.jai.2024.02.003}
\showDOI{\tempurl}


\bibitem[\protect\citeauthoryear{OpenAI}{OpenAI}{2023}]%
        {openai-chatgpt}
\bibfield{author}{\bibinfo{person}{OpenAI}.} \bibinfo{year}{2023}\natexlab{}.
\newblock \bibinfo{title}{ChatGPT}.
\newblock
\newblock
\urldef\tempurl%
\url{https://openai.com/chatgpt}
\showURL{%
\tempurl}


\bibitem[\protect\citeauthoryear{Panait and Luke}{Panait and Luke}{2005}]%
        {RN458}
\bibfield{author}{\bibinfo{person}{Liviu Panait} {and} \bibinfo{person}{Sean Luke}.} \bibinfo{year}{2005}\natexlab{}.
\newblock \showarticletitle{Cooperative Multi-Agent Learning: The State of the Art}.
\newblock \bibinfo{journal}{\emph{Autonomous Agents and Multi-Agent Systems}} \bibinfo{volume}{11}, \bibinfo{number}{3} (\bibinfo{year}{2005}), \bibinfo{pages}{387–434}.
\newblock
\showISSN{1387-2532}
\urldef\tempurl%
\url{https://doi.org/10.1007/s10458-005-2631-2}
\showDOI{\tempurl}


\bibitem[\protect\citeauthoryear{Pateria, Subagdja, Tan, and Quek}{Pateria et~al\mbox{.}}{2022}]%
        {RN230}
\bibfield{author}{\bibinfo{person}{Shubham Pateria}, \bibinfo{person}{Budhitama Subagdja}, \bibinfo{person}{Ah-hwee Tan}, {and} \bibinfo{person}{Chai Quek}.} \bibinfo{year}{2022}\natexlab{}.
\newblock \showarticletitle{Hierarchical Reinforcement Learning}.
\newblock \bibinfo{journal}{\emph{ACM computing surveys.}} \bibinfo{volume}{54}, \bibinfo{number}{5} (\bibinfo{year}{2022}), \bibinfo{pages}{1--35}.
\newblock
\showISSN{0360-0300}
\urldef\tempurl%
\url{https://doi.org/10.1145/3453160}
\showDOI{\tempurl}


\bibitem[\protect\citeauthoryear{Selmonaj, Szehr, Del~Rio, Antonucci, Schneider, and Rüegsegger}{Selmonaj et~al\mbox{.}}{2023}]%
        {RN89}
\bibfield{author}{\bibinfo{person}{Adrian Selmonaj}, \bibinfo{person}{Szehr Szehr}, \bibinfo{person}{Giacomo Del~Rio}, \bibinfo{person}{Alessandro Antonucci}, \bibinfo{person}{Adrian Schneider}, {and} \bibinfo{person}{Michael Rüegsegger}.} \bibinfo{year}{2023}\natexlab{}.
\newblock \showarticletitle{Hierarchical Multi-Agent Reinforcement Learning for Air Combat Maneuvering}. In \bibinfo{booktitle}{\emph{2023 International Conference on Machine Learning and Applications (ICMLA)}}. \bibinfo{publisher}{ICMLA}, \bibinfo{address}{Jacksonville}, \bibinfo{pages}{1031--1038}.
\newblock
\showISBNx{1946-0759}
\urldef\tempurl%
\url{https://doi.org/10.1109/ICMLA58977.2023.00153}
\showDOI{\tempurl}


\bibitem[\protect\citeauthoryear{Shehory and Kraus}{Shehory and Kraus}{1998}]%
        {RN460}
\bibfield{author}{\bibinfo{person}{Onn Shehory} {and} \bibinfo{person}{Sarit Kraus}.} \bibinfo{year}{1998}\natexlab{}.
\newblock \showarticletitle{Methods for task allocation via agent coalition formation}.
\newblock \bibinfo{journal}{\emph{Artificial Intelligence}} \bibinfo{volume}{101}, \bibinfo{number}{1} (\bibinfo{year}{1998}), \bibinfo{pages}{165--200}.
\newblock
\showISSN{0004-3702}
\urldef\tempurl%
\url{https://doi.org/10.1016/S0004-3702(98)00045-9}
\showDOI{\tempurl}


\bibitem[\protect\citeauthoryear{Sobol'}{Sobol'}{1967}]%
        {SOBOL196786}
\bibfield{author}{\bibinfo{person}{Ilya Sobol'}.} \bibinfo{year}{1967}\natexlab{}.
\newblock \showarticletitle{On the distribution of points in a cube and the approximate evaluation of integrals}.
\newblock \bibinfo{journal}{\emph{U. S. S. R. Comput. Math. and Math. Phys.}} \bibinfo{volume}{7}, \bibinfo{number}{4} (\bibinfo{year}{1967}), \bibinfo{pages}{86--112}.
\newblock
\showISSN{0041-5553}
\urldef\tempurl%
\url{https://doi.org/10.1016/0041-5553(67)90144-9}
\showDOI{\tempurl}


\bibitem[\protect\citeauthoryear{Strickland}{Strickland}{2022}]%
        {RN199}
\bibfield{author}{\bibinfo{person}{Laura~Gail Strickland}.} \bibinfo{year}{2022}\natexlab{}.
\newblock \emph{\bibinfo{title}{Coordinating Team Tactics for Swarm-vs.-Swarm Adversarial Games}}.
\newblock Thesis. \bibinfo{school}{School of Interactive Computing, College of Computing}.
\newblock


\bibitem[\protect\citeauthoryear{Trianni}{Trianni}{2008}]%
        {RN459}
\bibfield{author}{\bibinfo{person}{Vito Trianni}.} \bibinfo{year}{2008}\natexlab{}.
\newblock \bibinfo{booktitle}{\emph{Evolutionary Robotics for Self-Organising Behaviours}}.
\newblock \bibinfo{publisher}{Springer Berlin Heidelberg}, \bibinfo{address}{Berlin, Heidelberg}, \bibinfo{pages}{47--59}.
\newblock
\showISBNx{978-3-540-77612-3}
\urldef\tempurl%
\url{https://doi.org/10.1007/978-3-540-77612-3_4}
\showDOI{\tempurl}


\bibitem[\protect\citeauthoryear{Vezhnevets, Osindero, Schaul, Heess, Jaderberg, Silver, and Kavukcuoglu}{Vezhnevets et~al\mbox{.}}{2017}]%
        {RN276}
\bibfield{author}{\bibinfo{person}{Alexander~Sasha Vezhnevets}, \bibinfo{person}{Simon Osindero}, \bibinfo{person}{Tom Schaul}, \bibinfo{person}{Nicolas Heess}, \bibinfo{person}{Max Jaderberg}, \bibinfo{person}{David Silver}, {and} \bibinfo{person}{Koray Kavukcuoglu}.} \bibinfo{year}{2017}\natexlab{}.
\newblock \showarticletitle{Feudal networks for hierarchical reinforcement learning}. In \bibinfo{booktitle}{\emph{International conference on machine learning}}. \bibinfo{publisher}{PMLR}, \bibinfo{address}{Sydney}, \bibinfo{pages}{3540--3549}.
\newblock
\showISBNx{2640-3498}


\bibitem[\protect\citeauthoryear{Wang, Li, Gao, and Xie}{Wang et~al\mbox{.}}{2021}]%
        {RN91}
\bibfield{author}{\bibinfo{person}{Baolai Wang}, \bibinfo{person}{Shengang Li}, \bibinfo{person}{Xianzhong Gao}, {and} \bibinfo{person}{Tao Xie}.} \bibinfo{year}{2021}\natexlab{}.
\newblock \showarticletitle{UAV Swarm Confrontation Using Hierarchical Multiagent Reinforcement Learning}.
\newblock \bibinfo{journal}{\emph{International Journal of Aerospace Engineering}}  \bibinfo{volume}{2021} (\bibinfo{year}{2021}), \bibinfo{pages}{1--12}.
\newblock
\showISSN{1687-5966}
\urldef\tempurl%
\url{https://doi.org/10.1155/2021/3360116}
\showDOI{\tempurl}


\bibitem[\protect\citeauthoryear{Watson and Pollack}{Watson and Pollack}{2005}]%
        {RN461}
\bibfield{author}{\bibinfo{person}{Richard~A. Watson} {and} \bibinfo{person}{Jordan~B. Pollack}.} \bibinfo{year}{2005}\natexlab{}.
\newblock \showarticletitle{Modular Interdependency in Complex Dynamical Systems}.
\newblock \bibinfo{journal}{\emph{Artificial Life}} \bibinfo{volume}{11}, \bibinfo{number}{4} (\bibinfo{year}{2005}), \bibinfo{pages}{445–457}.
\newblock
\showISSN{1064-5462}
\urldef\tempurl%
\url{https://doi.org/10.1162/106454605774270589}
\showDOI{\tempurl}


\bibitem[\protect\citeauthoryear{Whitley}{Whitley}{1991}]%
        {Whitley1991}
\bibfield{author}{\bibinfo{person}{L.~Darrell Whitley}.} \bibinfo{year}{1991}\natexlab{}.
\newblock \showarticletitle{Fundamental Principles of Deception in Genetic Search}.
\newblock In \bibinfo{booktitle}{\emph{Foundations of Genetic Algorithms}}, \bibfield{editor}{\bibinfo{person}{Gregory~J.E. Rawlins}} (Ed.). Vol.~\bibinfo{volume}{1}. \bibinfo{publisher}{Elsevier}, \bibinfo{address}{San Mateo}, \bibinfo{pages}{221--241}.
\newblock
\showISSN{1081-6593}
\urldef\tempurl%
\url{https://doi.org/10.1016/B978-0-08-050684-5.50017-3}
\showDOI{\tempurl}


\bibitem[\protect\citeauthoryear{Yang, Borovikov, and Zha}{Yang et~al\mbox{.}}{2020}]%
        {RN191}
\bibfield{author}{\bibinfo{person}{Jiachen Yang}, \bibinfo{person}{Igor Borovikov}, {and} \bibinfo{person}{Hongyuan Zha}.} \bibinfo{year}{2020}\natexlab{}.
\newblock \showarticletitle{Hierarchical Cooperative Multi-Agent Reinforcement Learning with Skill Discovery}. In \bibinfo{booktitle}{\emph{International Foundation for Autonomous Agents and Multiagent Systems}}. \bibinfo{publisher}{International Foundation for Autonomous Agents and Multiagent Systems}, \bibinfo{address}{Auckland}, \bibinfo{pages}{1566--1574}.
\newblock
\showISBNx{1450375189}
\urldef\tempurl%
\url{https://doi.org/10.5555/3398761.3398941}
\showDOI{\tempurl}


\end{thebibliography}


\newpage
\appendix
\section{Supplementary Material}

\subsection{Formal Problem Definition}

\subsubsection{Agents and State Representation}

We simulate a two-dimensional open battlefield consisting of Blue defenders $B$, Red attackers $R$, and a stationary target $T$. Each agent $i$ at time $t$ is represented by:
$$
s^i_t = \big(x^i_t, y^i_t, \theta^i_t, v^i_t, \ell^i_t\big),
$$
where:  
\begin{itemize}
    \item $(x^i_t, y^i_t) \in \mathbb{R}^2$ is the agent position, 
    \item $\theta^i_t \in S^1$ is the heading (radians),
    \item $v^i_t \in [v^i_{\min}, v^i_{\max}]$ is the speed,
    \item $\ell^i_t \in \{0,1\}$ indicates if the agent is alive.
\end{itemize}

The full environment state at time $t$ is:
$$S_t = \{ s^i_t \mid i \in B \cup R \cup T \}$$

\subsubsection{Observation Space}

All agents have full observability in the current work:
$$O_t = S_t$$
although we note that partial observability and communication constraints will be incorporated in future work.

\subsubsection{Action Space} 
Actions are represented as:
    $$a^i = (\text{turn}^i, \text{accel}^i) \in [-1,1]^2$$

\subsubsection{Motion Dynamics}

Agent dynamics are updated at each time step $\Delta t$ according to:
\begin{align*}
\theta^i_{t+\Delta t} &= \theta^i_t + \omega_\text{max} \cdot \text{turn}^i \cdot \Delta t, \\
v^i_{t+\Delta t} &= \text{clip}\big(v^i(t) + \text{accel}^i \cdot \Delta t, v_\text{min}, v_\text{max}\big), \\
x^i_{t+\Delta t} &= x^i_t + v^i_{t+\Delta t} \cos \theta^i_{t(+\Delta t} \cdot \Delta t, \\
y^i_{t+\Delta t} &= y^i_t + v^i_{t+\Delta t} \sin \theta^i_{t+\Delta t} \cdot \Delta t.
\end{align*}

Intuitively, actions express steering and acceleration intent, while domain parameters ensure feasible turn and speed rates compatible with realistic UAV motion.

\subsubsection{Collisions and Removal}
Two agents $i$ and $j$ collide if:
$$\|(x^i_t,y^i_t) - (x^j_t,y^j_t)\| \le r^i + r^j$$
where $r^i$ is the physical radius of agent $i$. 

If a Blue and Red collide, both are destroyed:
$$\ell^i_{t+1} = \ell^j_{t+1} = 0$$

If a Red collides with a target, both are destroyed:
$$\ell^r_{t+1} = \ell^k_{t+1} = 0 \quad r\in R, k\in T$$

\subsubsection{Episode Commencement}

For the Target $T$, position is $(0, 0)$, speed $0$ and heading $0$.

For each red agent $r \in R$:
\begin{itemize}
    \item Position is randomly assigned around the target with radius $[0.75, 0.95]$,
    \item Speed is randomly assigned in range $[v_{min}, v_{max}]$ for red agents,
    \item Heading is toward the target.
\end{itemize}

For each blue agent $b \in B$:
\begin{itemize}
    \item Position is derived from the position of the warehouse in the agent's chromosome,
    \item Speed is $v_{min}$ for blue agents,
    \item Heading is randomly assigned in range $[0, 2\pi]$.
\end{itemize}

\subsubsection{Action Policies}

Each agent type has its own action policy.

Targets do not move, however for consistency they still require an action. Consequently, their action each time step is always $a^T_t = (0,0)$ and $v^T_{max}=0$.

Red agents follow a predefined sinusoidal trajectory toward the target. Each trajectory is randomly created with amplitude uniformly drawn from $[-0.2,0.2]$ and frequency drawn from $[1-4]$.

Blue agents, $b \in B$, each use a heuristic $H^b$ for their action policy assigned according to their chromosome. The heuristic takes as input the observation state $O_t$ expressed in relative terms for its agent, the agent's parameter $P^b$ as assigned by their chromosome and the current time step:
$$a^b_t = \pi^b_H(O^b_t, P^b, t)$$

The heuristic may or may not be deterministic.

\subsubsection{Episode Termination}

Episodes terminate, at $t^\star$, when any of:
$$\sum_{i\in B}\ell^i_t = 0 \quad \text{(all defenders destroyed)}$$
$$\sum_{i\in R}\ell^i_t = 0 \quad \text{(all attackers destroyed)}$$
$$\exists k\in T : \ell^k_t = 0 \quad \text{(a target is reached)}$$
$$t \ge T_{\max} \quad \text{(time limit)}$$

\subsubsection{Objective}

The Blue team objective is to prevent attackers from reaching targets and eliminate all Red agents. Blue receives a binary terminal payoff:
$$
U_B(S_{t^\star}) =
\begin{cases}
1 & \text{if } \sum_{i\in R} \ell^i_{t^\star}=0 \text{ and all targets remain alive},\\[4pt]
0 & \text{otherwise}.
\end{cases}
$$

We do not explicitly define, or need, a payoff for the Red team.

\subsubsection{Discussion}

This formulation makes explicit the controllable action space, motion dynamics, collision-induced elimination and episode termination conditions. It provides a precise grounding for discussing heuristic design, strategy composition, and the GA-DP optimization procedure in the main paper.

\subsection{Heuristics}

\begin{table}[htbp]
\caption{Suite of low-level policies ("heuristics").}
\label{tab:Heuristics}
\makegapedcells
\begin{tabular}{p{0.2\linewidth} p{0.70\linewidth}}
\toprule
Heuristic & Description \\\midrule
Straight & Fly straight using launch position and heading \\ 
Chase & Fly directly toward allocated attacker \\
Circle & Circle target with random radius\\
Predictive \newline Interception & Fly toward allocated attacker's straight-line predicted position \\
Rotation \newline Patrol & Fly in figure-eight pattern around target \\
Zone \newline Defense & Area around target divided into zones which allocated agents patrol \\
Defensive \newline Perimeter & Fly in either circle or polygon around target \\
Dynamic \newline Positioning & Defenders dynamically reposition based on the changing threat landscape \\
Threat Level \newline Assessment & Agents move toward attackers with closest proximity to and speed toward target \\
Adaptive \newline Patrol Radius & Agents patrol a circular path around the target, with radius based on attacker proximity \\
Threat-Weighted \newline Positioning & Agents position within an assigned zone, and shift toward the most threatening attacker \\
\bottomrule
\end{tabular}
\end{table}

Table \ref{tab:Heuristics} summarises the suite of low-level heuristic policies available to Blue agents. These heuristics constitute the action-generating primitives used throughout Stages 1–4, and form the building blocks from which full defender chromosomes are constructed

\subsection{Genetic Algorithm for Stage 1}

\begin{table}[h!]
\centering
\caption{GA parameters used for evolving small-scale Blue team strategies.}\label{tab:GAparams}
\begin{tabular}{lc}
\hline
Parameter & Value \\
\hline
Population size $N_\text{pop}$ & 1024 \\
Initialisation & Random \\
Generations $G$ & 128 \\
Elitism rate $r_e$& 0.2 \\
Single-point crossover rate $p_c$ & 0.7 \\
Mutation rate $p_m$ & 0.1 \\
Tournament size $k$ & 2 \\
\hline
\end{tabular}
\end{table}

Table \ref{tab:GAparams} summarises key parameters used in the Genetic Algorithm applied to small-scale team configurations in Stage 1. Testing of Sobol~\cite{SOBOL196786} and Latin Hypercube~\cite{LHC} sampling population initialisation methods showed immaterial difference in performance over purely random initialisation.

\subsection{Proof of Proposition 1}

\begin{proof}
The algorithm stores one entry for each pair of red and blue counts, giving $O(RB)$ entries. 
For a subgroup of $n \le R$ attackers, the algorithm enumerates all partitions of $n$ into groups of size at most $k$. 
The number of such restricted partitions grows as $O(n^{k-1})$ for fixed $k$~\cite{Nathanson2000}. 
Since $k$ is constant, this step is polynomial and dominated by the subsequent split factor. 
In the DP, enumerating possible splits of subgroups across $O(R^2)$ intervals contributes a factor of $O(R^2)$ overall. 
For each red split, defenders must be distributed across subgroups, which requires $O(B^2)$ time. 

For each DP entry, this leads to at most $O(R^2 B^2)$ work. 
Since there are $O(RB)$ entries in total, the overall runtime is $O(R^3 B^3)$. 
Because $k$ is fixed in advance, the partition step contributes only polynomial factors, and hence the overall runtime is polynomial in both $R$ and $B$.
\end{proof}

\begin{remark}
There are two distinct sources of polynomial overhead in the dynamic program. 
First, the enumeration of restricted partitions grows as $O(n^{k-1})$ for subgroups of size $n$, which remains polynomial for fixed $k$. 
Second, the enumeration of contiguous splits of attackers contributes $O(R^2)$ possibilities overall. 
Since the latter dominates asymptotically, the complexity is ultimately governed by the $O(R^3 B^3)$ bound.
\end{remark}

\subsection{Learning Dynamics Comparison}

\begin{figure}
    \centering
    \includegraphics[width=0.8\linewidth]{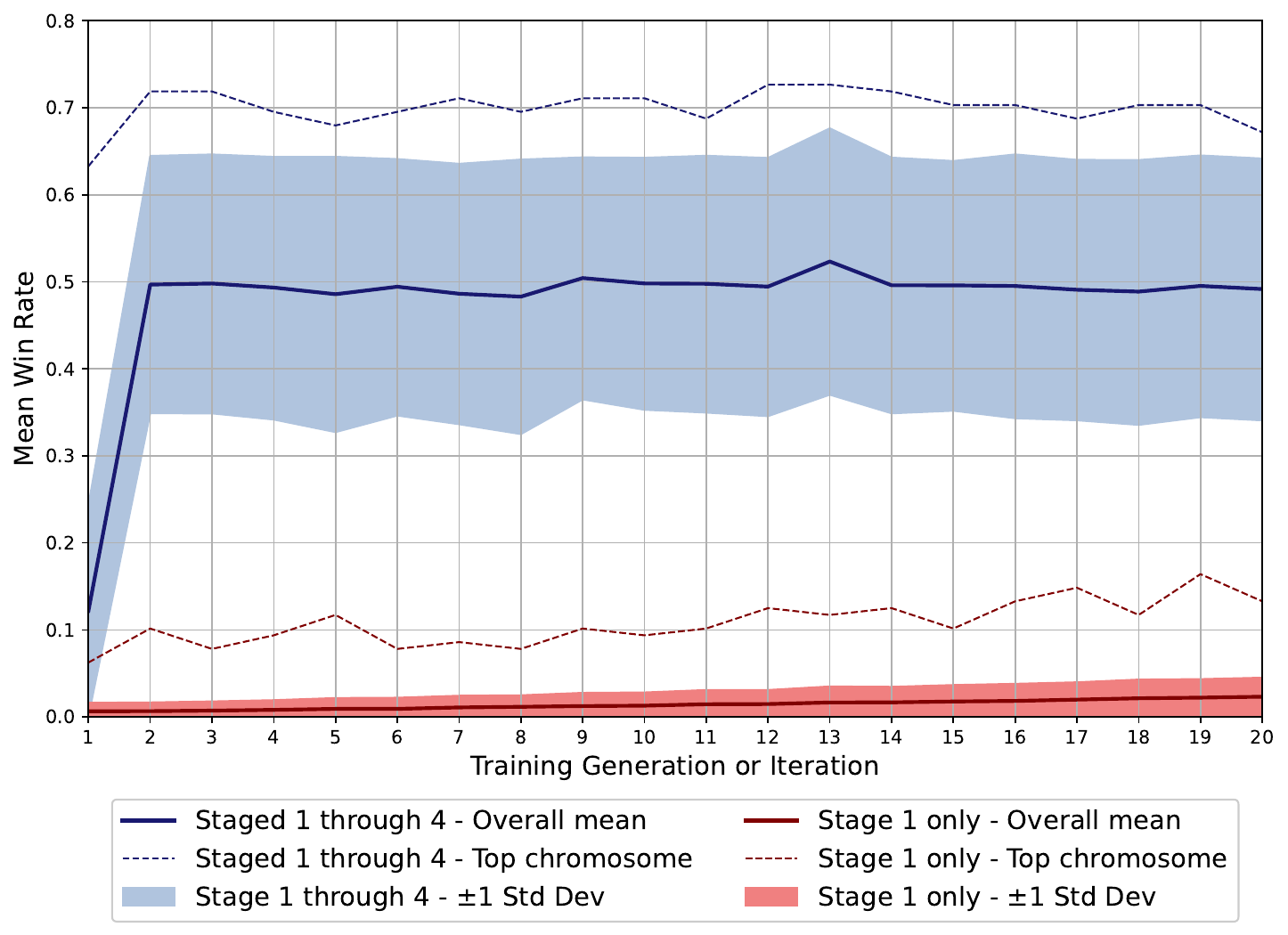}
    \caption{Learning dynamics in a 10R vs. 15B scenario. GA alone (Stage 1) exhibits gradual improvement over iterations, whereas the full Stage 1–4 pipeline rapidly exploits priors.}
    \label{fig:learning_comparison}
\end{figure}

Figure \ref{fig:learning_comparison} complements the ablation study by illustrating the learning dynamics in a moderate 10R vs. 15B scenario. Here, we compare the GA alone (Stage 1) with the full Stage 1–4 pipeline. The line plot shows that iterative Stage 1–4 learning rapidly improves win rates by exploiting priors from small-scale building blocks, whereas the GA alone exhibits slower convergence and lower final performance. This highlights that, beyond ablation effects, the staged pipeline accelerates the discovery of effective strategies during repeated simulations.

\subsection{Impact of Heuristic Composition on Performance}

To assess the contribution of individual heuristics, we performed pruning experiments on the evolved population. Removing the four least frequent heuristics produced negligible changes in performance across both small- and large-scale simulations, confirming redundancy. Removing the single most frequent heuristic in addition to the four least frequent heuristics substantially reduced performance: small-scale scenarios only achieved perfect win-rates under highly imbalanced conditions (for example, 1R vs. 4B), while competitive matchups saw dramatic drops. Large-scale simulations amplified this effect. These results are expected, as removing the key heuristic deprives the population of the primary mechanism for coordinated success. Critically, this demonstrates that the GA is robust in discovering effective strategies when high-impact heuristics are available, and the observed performance drop reflects sensitivity to the removal of essential building blocks rather than a limitation of the evolutionary process itself.

\onecolumn
\subsection{LLM Prompts}

To create a suite of candidate heuristics for the drone defense simulation, we drew on multiple large language models to generate diverse behavioural strategies. Specifically, we used ChatGPT-4o, ChatGPT-4.1 and ChatGPT-4 Turbo \cite{openai-chatgpt}, along with Gemini 2.0 Flash and Gemini 2.5 Flash \cite{gemini}. Using multiple models helped capture differences in reasoning styles, output formats, and suggested approaches, increasing the diversity of candidate heuristics.  A representative prompt used with ChatGPT-4o is shown in Figure \ref{fig:representativePrompt}.

\begin{figure*}[h]
\begin{tcolorbox}[drop shadow southeast, enhanced, sharp corners]
\begin{quote}
My simulation has blue quadcopters defending a target from red quadcopters. If quadcopters from opposing teams collide, they are both destroyed. If a red quadcopter reaches the target they then red wins. Blue must destroy all red quadcopters to win.
\end{quote}

\begin{quote}
I would like you to propose 10 different strategies for the blue team.
\end{quote}

\hrule
\vspace{1em}

Here are 10 different strategies for the blue team to defend the target and counter the red quadcopters:

\subsubsection{1. \textbf{Chase and Intercept}}
\begin{itemize}
    \item Each blue quadcopter targets a specific red quadcopter and chases it until it is destroyed.
    \item Prioritize the closest red quadcopter to minimize interception time.
\end{itemize}

\subsubsection{2. \textbf{Defensive Perimeter}}
\begin{itemize}
    \item Blue quadcopters form a circular or polygonal perimeter around the target.
    \item They patrol the perimeter and intercept any red quadcopter that breaches it.
\end{itemize}

\subsubsection{3. \textbf{Zone Defense}}
\begin{itemize}
    \item Divide the area around the target into zones, and assign each blue quadcopter to a specific zone.
    \item Blue quadcopters only engage red quadcopters that enter their assigned zones.
\end{itemize}

\subsubsection{4. \textbf{Swarm Behavior}}
\begin{itemize}
    \item Blue quadcopters move as a cohesive group (swarm) to overwhelm individual red quadcopters.
    \item Use swarm algorithms like Boids or Reynold's rules for coordinated movement.
\end{itemize}

\subsubsection{5. \textbf{Predictive Interception}}
\begin{itemize}
    \item Predict the future position of red quadcopters based on their current velocity and heading.
    \item Blue quadcopters move to intercept these predicted positions rather than chasing directly.
\end{itemize}

\paragraph{}
Each of these strategies can be implemented by modifying the behavior of the blue quadcopters in the \texttt{\_\_call\_\_} method of the \texttt{HeuristicChaseRed} class. Let me know if you'd like to explore or implement any of these strategies!
\end{tcolorbox}
\caption{Representative prompt to generate heuristic ideas using ChatGPT 4-0. The full Python simulation code was provided as context. The reference to the class at the bottom of the conversation is an error and should have referred to an abstract Blue class instead. Such errors had to be detected and manually corrected in any generated code.}
\label{fig:representativePrompt}
\end{figure*}

\newpage
\subsection{Generated Code Example}

Following ideation, the LLMs were requested to generate Python code to implement some of the strategies. The full package of existing Python simulation was provided as context and the \textit{BlueHeuristicAbstract} base class utilised to ensure API consistency. Ultimately each of these heuristics was manually converted to JAX so the simulations could be executed at scale on GPU. Example code generated by ChatGPT-4.1 is shown in Figure \ref{figCode1}.

\begin{figure*}[h]
    \begin{tcolorbox}[drop shadow southeast, enhanced, sharp corners]
\begin{lstlisting}[language=Python]
class HeuristicThreatWeightedPositioning(BlueHeuristicAbstract):
    def __init__(self, blue_idx, warehouse_idx, param_idx):
        super().__init__(blue_idx, warehouse_idx, param_idx)
        self.zone_idx = param_idx  # assign zone by param_idx

    def get_action(self, obs, state, blue_params, red_params, target_params):
        blue_pos = obs['blues'][self.blue_idx]
        target_pos = obs['targets'][0]
        reds_pos = obs['reds']
        # Divide area into zones (e.g., pie slices around target)
        num_zones = 5  # or set to number of blues
        angle_per_zone = 2 * np.pi / num_zones
        my_angle = self.zone_idx * angle_per_zone
        zone_center = target_pos + 0.5 
            * np.array([np.cos(my_angle), np.sin(my_angle)])

        # Find reds in my zone (by angle from target)
        reds_vecs = reds_pos - target_pos
        reds_angles = np.arctan2(reds_vecs[:,1], reds_vecs[:,0])
        reds_angles = (reds_angles + 2 * np.pi) % (2 * np.pi)
        zone_start = my_angle
        zone_end = my_angle + angle_per_zone
        in_zone = (reds_angles >= zone_start) & (reds_angles < zone_end)
        reds_in_zone = reds_pos[in_zone] if np.any(in_zone) else []

        # If reds in zone, move toward the one closest to target (threat-weighted)
        if len(reds_in_zone) > 0:
            dists = [np.linalg.norm(red - target_pos) for red in reds_in_zone]
            threat_red = reds_in_zone[np.argmin(dists)]
            weighted_pos = 0.7 * threat_red + 0.3 * zone_center
        else:
            weighted_pos = zone_center

        direction = weighted_pos - blue_pos
        if np.linalg.norm(direction) > 1e-6:
            direction = direction / np.linalg.norm(direction)
        else:
            direction = np.zeros_like(direction)
        return direction * blue_params.max_speed
\end{lstlisting}
    \end{tcolorbox}
    \caption{Example LLM-guided (ChatGPT-4.1) code for the Threat-Weighted Positioning heuristic. Code comments provided by the LLM.} 
    \label{figCode1}
\end{figure*}


\end{document}